%% file: around_fair_words.tex
\documentclass[11pt]{article}

\usepackage{amsopn}
\usepackage{amsfonts} 
\usepackage{amsmath}
\usepackage{amsthm}
\usepackage{amssymb} 
\usepackage{amscd}
\usepackage{xspace}
\usepackage{fullpage}
\usepackage[numbers]{natbib}
\usepackage{hyperref}
\hypersetup{pdfborder=0 0 0}

\newcommand{\N}{\ensuremath{\mathbb{N}}\xspace}
\newcommand{\init}{{\rm init}}
\newcommand{\final}{{\rm final}}

\newcommand{\pref}{{\rm pref}}
\newcommand{\suff}{{\rm suff}}
\newcommand{\rect}{{\rm Rect}}
\newcommand{\first}{{\rm first}}
\newcommand{\last}{{\rm last}}
\newcommand{\leftPart}{{\rm Left}}
\newcommand{\rightPart}{{\rm Right}}
\newcommand{\ca}{\textcolor{red}{a}}
\newcommand{\cb}{\textcolor{blue}{b}}

\newcommand{\pmat}{\text{pMat}}
\newcommand{\parikhMat}{\text{ParMat}}
\newcommand{\partition}{\text{Part}}

\theoremstyle{definition}
\newtheorem{definition}{Definition}[section]
\newtheorem{theorem}[definition]{Theorem}
\newtheorem{corollary}[definition]{Corollary}
\newtheorem{fact}[definition]{Fact}
\newtheorem{proposition}[definition]{Proposition}
\newtheorem{remark}[definition]{Remark}
\newtheorem{lemma}[definition]{Lemma}

\newcommand{\cerny}{\v{C}ern\'y}
\usepackage{tikz}
\usetikzlibrary{arrows,automata}
\newcommand*{\doi}[1]{\href{http://dx.doi.org/#1}{doi: #1}}
\begin{document}
\title{On some $2$-binomial coefficients of binary words: geometrical interpretation, partitions of integers, and fair words}
\author{G. Richomme}
\author{Gwenaël Richomme\\
LIRMM, Université de Montpellier, CNRS, Montpellier, France\\
and Groupe AMIS, Université de Montpellier Paul-Valéry, Montpellier, France}

\date{\today}
\maketitle

\begin{abstract}
The binomial notation $\binom{w}{u}$ represents the number of occurrences of the word $u$ as a (scattered) subword in $w$.
We first introduce and study possible uses of
a geometrical interpretation of $\binom{w}{ab}$
and $\binom{w}{ba}$ when $a$ and $b$ are distinct letters.
We then study the structure of 
the $2$-binomial equivalence class of a binary word $w$ (two words are $2$-binomially equivalent if
they have the same binomial coefficients, that is, the same numbers of occurrences, for each word of length at most $2$). 
Especially we prove the existence of an isomorphism between
the graph of the $2$-binomial equivalence class of $w$ with respect to a particular
rewriting rule and the lattice of partitions of the integer $\binom{w}{ab}$
with $\binom{w}{a}$ parts and greatest part bounded by $\binom{w}{b}$.
Finally we study binary fair words, the words over $\{a, b\}$ having
the same numbers of occurrences of $ab$ and $ba$
as subwords ($\binom{w}{ab}=\binom{w}{ba}$). In particular, we prove a recent conjecture related to a special case of the least square approximation.
\end{abstract}


\tableofcontents

\section{Introduction}

A word $u$ is a (scattered) subword of a word $w$
if it occurs as a subsequence of $w$.
For instance, the word ``see" is a subsequence
of the word ``subsequence" = ``s$\cdot$ubs$\cdot$e$\cdot$qu$\cdot$e$\cdot$nce". Actually the word ``see" appears in 6 different ways as a subsequence of the word ``subsequence". 
The binomial notation $\binom{w}{u}$, 
that we call a \textit{binomial coefficient of words},  
represents the number of occurrences of the word $u$ as a subword in $w$: $\binom{subsequence}{see}=6$. 
Basic properties of these coefficients can be found, for instance, in \cite[Chapter 6]{Lothaire1983book}.
Subwords and their numbers of occurrences
are central in many studies. See for instance 
\cite{Fleischmann_Haschk_Hofer_Huch_Mayrock_Nowotka2023TCS,
Leroy_Rigo_Stipulanti2016AAM,
Lothaire1983book,
Lu_Chen_Wen_Wu2024TCS,
Manvel_Meyerowitz_Schwenk_Smith_Stockmeyer1991DM,
Mateescu_Salomaa_Salomaa_Yu2001TIA,
Pin_Silva2014EJC,
Prodinger1979DM,
Richomme_Rosenfeld2023STACS,
Rigo_Stipulanti_Whiteland2024EJC}
and their references therein.
Some generalizations have recently been considered (see \cite{Renard_Rigo_Whiteland2025JAC,Renard_Rigo_Whiteland2025DM}, \cite{Golm_Nahvi_Gabrys_Milenkovic2022ISIT,Rigo_Stipulanti_Whiteland2024ISIT}
and \cite{Golafshan_Rigo2025words}).

Given an integer $k$,
two words $u$ and $v$ are $k$-\textit{binomially equivalent}
if they have the same number of occurrences for all subwords of length at most $k$. 
Terminology ``$k$-binomial equivalence"
was introduced by M.~Rigo and P.~Salimov in 2015 \cite{Rigo_Salimov2015TCS}.
In their ``acknowledgments" part, 
they attribute 
to M.~Rigo and J.-E.~Pin the independent idea of this binomial equivalence.
The $k$-binomial equivalence is denoted $\sim_k$ in this paper.
This relation should not be confused with the Simon's congruence
also often denoted $\sim_k$ and also related to subwords:
two words $u$ and $v$ are equivalent for the Simon's congruence
if they have the same set of subwords of length at most $k$ 
(the number of occurrences of the subwords is not considered).
This congruence is often considered in reconstruction problems
(see, for instance, \cite{Dudik_Schulman2003JCTA,
Lothaire1983book,
Manvel_Meyerowitz_Schwenk_Smith_Stockmeyer1991DM,
Richomme_Rosenfeld2023STACS,Salomaa2003BEATCS} and their references therein).

It may be observed 
that any two words in an equivalence class of the relation
$\sim_k$ have the same lengths. In \cite{Rigo_Salimov2015TCS}, 
the number of classes in 
$A^n/\!\!\sim_2$ is given and for arbitrary integer $k$
the growth of the number of classes in 
$A^n/\!\!\sim_k$ is estimated. This paper has also introduced the notion of binomial complexity and was the starting point of some studies \cite{Lejeune_Leroy_Rigo2020JCTA,
Lu_Chen_Wen_Wu2024TCS,
Rigo_Stipulanti_Whiteland2024EJC}.

Among results on binomial equivalence,
let us mentioned that M.~Lejeune, M.~Rigo and M.~Rosenfeld \cite{Lejeune_Rigo_Rosenfeld2020IJAC} have proved that, given any alphabet $A$, the monoid
$A^*/\!\!\sim_2$ is 
isomorphic to the submonoid, generated by $A$ of the nil-$2$ group.
Even if the terminology ``binomial equivalence" was introduced in 2015, this relation was previously studied in particular cases and
some results concerning the $2$-binomial equivalence 
have been published earlier.
For instance, in 2003,  M.~Dudik and L.~J~Schulman 
have already denoted it $\sim_k$ in \cite{Dudik_Schulman2003JCTA} without naming it.
In 2001, 
A.~Mateescu, A.~Salomaa, J.~Salomaa and S.~Yu 
introduced the notion of Parikh matrix.
This data structure allows to store some binomial coefficients of a word.
These coefficients depend on the considered alphabet.
In the binary case, 
the matrices stores three coefficients that are sufficient to know
all the number of occurrences of subwords of length $2$.
A lot of studies around Parikh matrices
concern the study of words having a same given matrix. 
In this context, two words are Parikh-equivalent if they have the
same Parikh matrix. In the binary case,
two words are Parikh equivalent if and only if
they are $2$-binomially equivalent words.
Hence in the binary case, any result on Parikh equivalence (see, for instance,
\cite{Atanasiu_Martin-Vide_Mateescu2002FI,Fosse_Richomme2004IPL})
is a result 
on $2$-binomial equivalence. 
Maybe the first result on this vein is Theorem~2 in \cite{Prodinger1979DM} (see Theorem~\ref{T_Prodinger_syntactic_congruence_extended}), a 1979 paper in which the Parikh matrix appears, without being named and in the binary case only, more than twenty years before the independent introduction of this notion. In the Parikh matrices literature, 
this paper seems to be mentioned only in 2008 by A.~\cerny\ \cite{Cerny2008IJFCS}.
Concerning the $k$-binomial equivalence,
it is worth observing that in \cite{Cerny2009JALC} A.\cerny\ has introduced
the notion of precedence matrix that stores all the information needed to 
know all the number of occurrences of subwords of length $2$ for an arbitrary alphabet.

This paper of A.\cerny\ \cite{Cerny2009JALC}, written in 2006 and published in 2009, also introduced the notion of \textit{fair words} as the words having the same numbers of occurrences of the subwords $ab$ and $ba$ for all distinct letters $a$ and $b$. He made a conjecture about the number of fair words of length $n$.
Actually a 1979 result by H.~Prodinger \cite{Prodinger1979DM} has provided a formula for the growth of binary fair words, showing that the A.~{\cerny} is false. A.~{\cerny} himself cited
this Prodinger's article in the paper \cite{Cerny2008IJFCS}, written in 2007 and published in 2008. 

\medskip

\textbf{Contributions of this paper}

In \cite{Prodinger1979DM}, H.~Prodinger's explains that 
his approach of fair words 
was motivated by generalizing some aspects of the extended Dyck Language. 
Since Dyck words are usually graphically represented, 
the natural question of the existence of a graphical representation for fair words, and more
generally for some binomial coefficient of words, raises.
Although it seems to be unknown, 
considering a usual graphical representation of a binary word $w$,
the coefficients $\binom{w}{ab}$ and $\binom{w}{ba}$ appear 
to correspond to the areas of two complementary parts of the rectangle
of width $\binom{w}{a}$ (the number of $a$ occurring in $w$) and of height
$\binom{w}{b}$. After recalling in Section~\ref{subsec_bases} this usual graphical representation of a binary word,
Section~\ref{subsec_Interpret} presents the interpretation
of $\binom{w}{ab}$ and $\binom{w}{ba}$ for words $w$ over arbitrary alphabets.
Using a link established in
\cite{Fosse_Richomme2004IPL} 
between the number $\binom{w}{ab}$ of occurrences
of the subword $ab$ in a binary word and the sum
of position of the occurrences of the letter $b$ in $w$,
an interpretation of this sum is provided
using another graphical representation of a binary word.

The aim of Section~\ref{sec_binom_equiv} is to provide examples
of uses of the first geometrical interpretation of $\binom{w}{ab}$.
This is done interpreting some
formulas and a known result, initially stated in the context of
Parikh matrices, 
that characterizes the $2$-binomial equivalence by
another equivalence based on a rewriting rule on words.
Section~\ref{subsec_def_binom_equiv} recalls the definition
of the $k$-binomial equivalence and basic properties.
Section~\ref{subsec_storing} explains with further details
the natural links existing between the
$2$-binomial equivalence and the Parikh and precedence matrices.
In Section~\ref{subsec_interpret_char}, we recall the characterization of the $2$-binomial equivalence 
of binary words using a rewriting rule. We explain how this rule works 
graphically and how the characterization can be proved graphically.
We end formalizing this new proof.

Section~\ref{secPartitions}
considers the structure of equivalence classes for the 
$2$-binomial equivalence in the binary case.
First, in Section~\ref{subsec_first_lattice},
this structure is analysed using the relation induced
by the rewriting rule mentioned in the previous section.
It is proved that
the graph of the relation is a lattice whose
least and greatest elements are characterized.
Remark~\ref{rem_bibli_init_final} presents two previous
uses of these natural representants of $2$-binomial equivalence classes.
In Section~\ref{subsec_bij_lattice_integers}, 
we exhibit another relation using a rewriting rule.
The graph of this relation for the
$2$-binomial equivalence class of a word $w$
is isomorphic to a part of the lattice of the partition
of the integer $\binom{w}{ab}$.
It is worth noting that without this structural aspect,
a link between elements of a $2$-binomial equivalence classes and partitions of integers have already been observed
\cite{Mateescu_Salomaa2004IJFCS,
Mathew_Thomas_Bera_Subramanian2019AAM,
Teh_Kwa2015TCS,Teh2015IJFCS}).

Section \ref{sec_fair_words}
is dedicated to the study of fair words mostly in the binary case.
Graphically a fair word $w$ corresponds
to a word whose graphical representation divides the rectangle of width $\binom{w}{a}$ and height $\binom{w}{b}$ in two parts of same area. A fair word is also a word such that
$\binom{w}{ab} = \frac{1}{2}\binom{w}{a}\binom{w}{b}$.
In Section~\ref{subsec_fair_words_base}, 
basic notions on fair words are presented. These leads
to see that a word is fair if and only the least and greatest elements of the previous lattice is a palindrome.
A new characterization 
of the $2$-binomial equivalence is obtained using 
rewriting rules inspired by palindromic amiability \cite{Atanasiu_Martin-Vide_Mateescu2002FI} but 
replacing palindrome with fair words.
Section~\ref{subsec_language_properties}
extends to arbitrary alphabets a
Prodinger's characterization of the $2$-binomial equivalence
originally stated for binary alphabet.
Section~\ref{subsecSturm}
considers fair balanced words.
Balanced words are factors of Sturmian words 
that are known to correspond to digitalized straight lines. 
Hence balanced words represents segments.
A natural question is thus
the existence of fair balanced words:
their natural graphical representation would correspond
to a segment digitalization separating the rectangle of the
representation in two parts of same area.
Answering positively, we show that
for any fair word $w$, there exists a $2$-binomially
equivalent balanced word and, moreover, we prove that
a balanced word is fair if and only if it is a palindrome.
Section~\ref{subsec_number_fair_words}, the last part
of Section~\ref{sec_fair_words}, is concerned
by the numbers of fair words of each length.
The first values of these numbers have been
computed by A.~\cerny\ \cite{Cerny2009JALC}. 
This sequence is already known in the 
On-Line Encyclopedia of Integers \cite{OEISA222955} (sequence A222955)
but enumerating another family of words related to
the least squares optimization method. We prove that this
family of words is also the family of fair words. 
Using once again the connection between $\binom{w}{ab}$
and the sum of position of $b$s in $w$, we also solve a
2023 conjecture presented in \cite{OEISA222955}.

In Section~\ref{sec_conclusion}, 
we conclude. 
First, we summarize all characterizations of
the $2$-binomial equivalence recalled or proved in the paper.
Secondly, we present several questions arising from the observation that fair words may be considered as a generalization
of palindromes since palindromes are fair.

\section{\label{secInterpret}Geometrical interpretations of two binomial coefficients of words}

\subsection{\label{subsec_bases}Basic notions and a geometrical representation of a binary word}

We assume that readers are familiar with Combinatorics on Words (see for instance \cite{Lothaire1983book,Lothaire2002book})
but we need to specify some notation.
Let $w = w_1 \cdots w_n$ be a word over an alphabet $A$: 
$w_i \in A$ for $1 \leq i \leq n$.
The length of $w$, denoted $|w|$, is the integer $n$.
The \textit{empty word}, denoted $\varepsilon$, is the word of length $0$.
The \textit{mirror image} or \textit{reverse} of $w$ is the word
$\widetilde{w} = w_n \cdots w_1$.

A word $u$ is a \textit{factor} of a word $w$ if there exist words $p$ and $s$ such
that $w = pus$. 
If $p = \varepsilon$, $u$ is a \textit{prefix} of $w$.
If $s = \varepsilon$, $u$ is a \textit{suffix} of $w$.
If $u \neq w$, $u$ is a \textit{proper} factor, prefix or suffix of $w$.
For any integer $i$, $0 \leq i \leq |w|$, 
let $\pref_{i}(w)$ be the prefix of length $i$ of $w$ 
and let $\suff_{i}(w)$ be the prefix of length $i$ of $w$:
$\pref_0(w) = \suff_0(w) = \varepsilon$ and, for $i \geq 1$,
$\pref_i(w) = w_1 \cdots w_i$, 
$\suff_i(w) = w_{|w|-i+1}\cdots w_{|w|}$.

For any set $S$, $\#S$ denotes its cardinality.
Let $A$ be an alphabet.
Given a letter $\alpha \in A$, 
$|w|_\alpha$ denotes the number of occurrences of $a$ in $w$ that is
$\#\{ i \mid w_i = \alpha, 1 \leq i \leq n\}$.
Given a letter $\alpha \in A$
and an integer $i$,  $1 \leq i \leq |w|_\alpha$,
let $\pref_{i,\alpha}(w)$ denote the prefix $p$
of $w$ ending with $\alpha$ and containing 
exactly $i$ occurrences of $\alpha$ ($|p|_\alpha = i$).
Also let $\suff_{i,\alpha}(w)$ denote the suffix $s$
of $w$ beginning with $\alpha$ and containing 
exactly $i$ occurrences of $\alpha$ ($|s|_\alpha = i$).
For instance,
for $w = aabaaba$, $\pref_{2,a}(w) = aa$
and $\suff_{2,a}(w) = aba$.

Let $A = \{a_1, \ldots, a_k\}$ be a $k$-letter alphabet.
Assume $A$ is ordered with $a_1 < a_2 < \cdots < a_k$.
For any word $w$, the \textit{Parikh vector} of $w$ with respect to this order
is the $k$-uple $(|w|_{a_1}, \ldots, |w|_{a_k})$, denoted $\Psi(w)$.
In the rest of the paper, we will essentially
consider binary alphabets $\{a, b\}$ and $\{0, 1\}$
with the natural orders $a < b$ and $0 < 1$.
The lexicographical order on $A^*$ that extends the order $<$ will be denoted $<_{lex}$.

The sequence $(\Psi(\pref_i(w))_{0 \leq i \leq |w|}$ 
can be interpreted as a sequence of points in the grid $\N^k$ and,
joining the neighbour points, 
the sequence $((\Psi(\pref_{i-1}(w)),\Psi(\pref_{i}(w)))_{1 \leq i \leq |w|}$
is a usual representation of $w$ has a broken line that we call the \textit{line representation} of $w$.
Figure~\ref{fig_aabaabbababba}
provides this line representation of the binary word $aabaabbababba$.
The letter $a$ is represented by an horizontal segment and the letter $b$ is represented by a
vertical line. The line representation of $w$ joins the point $(0,0)$ to the point $(|w|_a, |w|_b) = (7,6)$ in the grid.

\input{grille_6x7.tex}

\subsection{\label{subsec_Interpret}Geometrical representations of 
\texorpdfstring{$\binom{w}{ab}$ and $\binom{w}{ba}$}{two binomial coefficients} }

A word $u = u_1 \cdots u_k$ (with $u_i$ a letter for each $i$, $1 \leq i \leq k$)
is a (scattered) \textit{subword}
of a word $w$ if there exist words
$(x_i)_{0 \leq i \leq k}$
such that $w = x_0 u_1 x_1 \cdots u_k x_k$.
The increasing sequence
$(|x_0 u_1 x_1 \cdots x_{i-1}u_i|)_{1 \leq i \leq k}$
denotes an \textit{occurrence} of $u$ in $w$. 
The number of occurrences of $u$ as a subword in $w$
is denoted $\binom{w}{u}$. 
This classical notation as a binomial coefficient of words is
inspired by the fact that, for any letter $\alpha$ and any integers $n$ and $k$, it holds
$\binom{a^n}{a^k} = \binom{n}{k}$.
By convention $\binom{w}{\varepsilon} = 1$.
For example, $\binom{aabaaba}{aba} = 10$.
Note that for any letter $\alpha$, $\binom{w}{\alpha} = |w|_\alpha$.

Let $a$, $b$ be two different letters of an alphabet $A$.
Consider an occurrence $(k, \ell)$ of $ab$ in a word $w$ over $A$.
There exist words $x$ and $y$
such that $xayb$ is a prefix of $w$
with $|xa|=k$ and
$|xayb|=\ell$.
Observe that $xa = \pref_k(w)$ and $xayb = \pref_\ell(w)$.
Let $i = |xa|_a$
and $j = |xayb|_b$.
Observing that $xayb = \pref_{j, b}(w)$, we note that
we have associated with the occurrence of the subword $ab$
an element of the set
$$\leftPart(w) = \{ (i, j) \mid 1 \leq j \leq |w|_b, 1 \leq i \leq |\pref_{j, b}(w)|_a\}.$$

Conversely given any $(i, j)$ in $\leftPart(w)$,
there exist unique words $x$ and $y$
such that $\pref_{j, b}(w) = xayb$ (in particular $|xayb|_b = j$) with $|xa|_a = i$.
To the element $(i, j)$ we have associated the occurrence
$(|xa|, |xayb|)$ of $ab$ with a subword of $w$.
Hence there is a natural bijection between occurrences of $ab$ as a subword of $w$ and
elements of $\leftPart(w)$.
This proves:

\begin{lemma}
\label{L_interpretation_binom_ab}
For any word $w$ over an alphabet 
with $\{a, b\} \subseteq A$, $a \neq b$,
$$\binom{w}{ab} = \#\leftPart(w)$$
\end{lemma}

Let us denote $\rect(w)$ the rectangle of width $|w|_a$ and
height $|w|_b$:
$$\rect(w) = \{ (i, j) \mid 1 \leq i \leq |w|_a, 1 \leq j \leq |w|_b\}.$$
The previous lemma supports the \textit{geometrical interpretation} of $\binom{w}{ab}$ 
as the area of the part of $\rect(w)$ which is at the left of the line representation of $\pi_{a,b}(w)$ where 
$\pi_{a,b}$ is the \textit{projection} morphism from $A^*$
to $\{a,b\}^*$ defined by
$\pi_{a,b}(a)=a$, $\pi_{a,b}(b)=b$ and $\pi_{a,b}(c)=\varepsilon$
for $c \in A\setminus\{a,b\}$
(term ``left" should be understood while drawing the broken line from point $(0,0)$ to point $(|w|_a, |w|_b)$). 
See Figure~\ref{fig_area}. At first observations, we can verify the well-known facts: for any word $w$, $\binom{w}{ab} \leq |w|_a|w|_b$~; for any $i$, $0 \leq i \leq |w|_a|w|_b$, there exists a word $w$ with $\binom{w}{ab} = i$.

\input{grille_6x7_area_formule_base_avec_longueurs.tex}

Let $\rightPart(w) = \{ (i, j) \mid 1 \leq j \leq |w|_b, |\pref_{j, b}(w)|_a + 1 \leq i \leq |w|_a\}$. With a similar reasoning as previously, one can get the following lemma (see also Figure~\ref{fig_area}).

\begin{lemma}
\label{L_interpretation_binom_ba}
For any word $w$ over an alphabet 
with $\{a, b\} \subseteq A$, $a \neq b$,
$$\binom{w}{ba} = \#\rightPart(w)$$
\end{lemma}

It follows the definitions that $\leftPart(w)$ and $\rightPart(w)$
form a partition of the rectangle $\rect(w)$. Hence, as presented by Figure~\ref{fig_area}, we get the well-known relation:

\begin{equation}\label{eq_relation_base1}
\binom{w}{ab}+ \binom{w}{ba} = |w|_a \times |w|_b 
\end{equation}

Let recall that $\leftPart(w) = \{ (i, j) \mid 1 \leq j \leq |w|_b, 1 \leq i \leq |\pref_{j, b}(w)|_a\}$. In this definition the values $\pref_{j, b}(w)$ describe all prefixes $pb$ of $w$. Since $|pb|_a = |p|_a$, Lemma~\ref{L_interpretation_binom_ab} can be translated to the next relation:

\begin{equation}\label{eq_calcul_binom_ab_par_prefixes}
\binom{w}{ab} = \sum_{pb \text{~prefix of~} w} |p|_a 
\end{equation}

This relation corresponds to the computation of $\#\leftPart(w)$ adding the numbers of squares occurring at the left of each $b$ in $\rect(w)$ (some figures are given in the appendix to support this idea). Of course, one can compute $\#\leftPart(w)$ adding the numbers of squares occurring above each $a$ in $\rect(w)$. 
We can observe that
$\leftPart(w) = \{ (i, j) \mid 1 \leq i \leq |w|_a, |\pref_{i,a}(w)|_b+1 \leq j \leq |w|_b\}$ or, since $|w|_b-|\pref_{i,a}(w)|_b = |\suff_{i,a}(w)|_b$,
$\leftPart(w) = \{ (i, |w|_b-j+1) \mid 1 \leq i \leq |w|_a, 
1 \leq j \leq |\suff_{i,a}(w)|_b\}$. In other terms,

\begin{equation}\label{eq_calcul_binom_ab_par_sufixes}
\binom{w}{ab} = \sum_{as \text{~suffix of~} w} |s|_b
\end{equation}

We can get similar formulas for $\#\rightPart(w)$ and $\binom{w}{ba}$. In particular
$\displaystyle
\binom{w}{ba} = \sum_{bs \text{~suffix of~} w} |s|_a$
and
$\displaystyle
\binom{w}{ba} = \sum_{pa \text{~prefix of~} w} |p|_b$

\subsection{\label{subsec_sum_positions}A link with the sum of positions of the letter \texorpdfstring{$b$}{b}}

Formula~\eqref{eq_calcul_binom_ab_par_prefixes} is linked to the formula given in the next lemma.
Let $S_b(w)$ be the sum of positions of the occurrences of the letter $b$ in the word $w$: $S_b(w) = \sum_{pb \text{~prefix of~} w} |pb|$.

\begin{lemma}[\cite{Fosse_Richomme2004IPL}]
\label{L_lien_Sb_et_binom_ab}
Given a word $w$ over an alphabet $A$ and given a letter $b$ of $A$,
\begin{equation}
S_b(w) = \frac{|w|_b(|w|_b+1)}{2} + \sum_{a \in A\setminus\{b\}} \binom{w}{ab} \label{eq_lien_Sb_et_binom_ab}
\end{equation}
\end{lemma}

The proof of this lemma is a direct consequence of formula~\eqref{eq_calcul_binom_ab_par_prefixes}. 
Since $|pb| = \sum_{a \in A} |pb|_a$ for any word $p$, 
let us observe that:
$$S_b(w) = \sum_{a \in A\setminus\{b\}}\sum_{pb \text{~prefix of~} w} |pb|_a +\sum_{pb \text{~prefix of~} w} |pb|_b$$
By Formula~\eqref{eq_calcul_binom_ab_par_prefixes}, 
the first sum is $\sum_{a \in A\setminus\{b\}}\binom{w}{ab}$. The second sum is $\sum_{i = 1}^{|w|_b} i = \frac{|w|_b(|w|_b+1)}{2}$.

Let $S_b'(w) = \sum_{bs \text{~suffix of~} w} |bs|$.
Observe that $S_b'(w) = S_b(\widetilde{w})$ and $\binom{w}{ba} = \binom{\widetilde{w}}{ab}$. Thus the following result is a corollary of Lemma~ \ref{L_lien_Sb_et_binom_ab}.

\begin{lemma}
\label{L_lien_Sb'_et_binom_ab}
Given a word $w$ over an alphabet $A$ and given a letter $b$ of $A$,
\begin{equation}
S_b'(w) = \frac{|w|_b(|w|_b+1)}{2} + \sum_{a \in A\setminus\{b\}} \binom{w}{ba} \label{eq_lien_Sb'_et_binom_ba}
\end{equation}
\end{lemma}

Consider a binary word $w $ over $\{a,b\}$.
Summing Relations~\eqref{eq_lien_Sb_et_binom_ab}
and \eqref{eq_lien_Sb'_et_binom_ba}, we get
$S_b(w) + S_b'(w) = |w|_b(|w|_b+1) + \binom{w}{ab} + \binom{w}{ba}$. 
From $|w|= |w|_a+|w|_b$ and relation~\eqref{eq_relation_base1}, we deduce:

\begin{equation}
S_b(w) + S_b'(w) = |w|_b(|w|+1) \label{eq_lien_entre_sommes}
\end{equation}

\input{interpretation_sommes.tex}

Figure~\ref{fig_area_sums} provides an illustration of a geometrical interpretation of this Formula~\eqref{eq_lien_entre_sommes}. 
We will not formalize this interpretation. 
But let us explain it. 
Here a binary word over $\{a, b\}$ 
is interpreted replacing occurrences of the letter $a$
with an horizontal segment as in the previous line interpretation
but replacing occurrences of the letter $b$
with a diagonal segment. 
As a consequence for the $i$th $b$ ($1 \leq i \leq |w|_b$), comparatively with the previous representation of a word, the 
area of the left part increases by $i-1$ squares and an half
while area of the right part increases by $|w|_b-i$ squares and an half.
Globally, both the left and right parts of the rectangle
increase by $\frac{|w|_b^2}{2}$.
So these left and right areas are respectively 
$\frac{|w|_b^2}{2}+\binom{w}{ab}$ and 
$\frac{|w|_b^2}{2}+\binom{w}{ba}$, or, 
by Equations~\eqref{eq_lien_Sb_et_binom_ab}
and \eqref{eq_lien_Sb'_et_binom_ba}, 
$S_b(w)-\frac{|w|_b}{2}$ and
$S_b'(w)-\frac{|w|_b}{2}$. 
The area of the whole rectangle is $|w|\times|w|_b$ (the rectangle has width $|w|$ and height $|w|_b$).
Since it is also $(S_b(w)-\frac{|w|_b}{2})+(S_b'(w)-\frac{|w|_b}{2})$,
we get Formula~\eqref{eq_lien_entre_sommes}.

\section{\label{sec_binom_equiv}About the \texorpdfstring{$2$}{2}-binomial equivalence}

\subsection{\label{subsec_def_binom_equiv}The \texorpdfstring{$k$}{k}-binomial equivalence}

Let $k$ be an integer. 
The $k$-\textit{binomial equivalence} \cite{Rigo_Salimov2015TCS} over an alphabet $A$ 
is defined as follows: 
for two words $u$ and $v$, $u \sim_k v$
if and only if $\binom{u}{x} = \binom{v}{x}$ 
for all words $x$ with $|x| \leq k$.
As shown by \cite[Lemma 1]{Manvel_Meyerowitz_Schwenk_Smith_Stockmeyer1991DM}, condition
$|x| \leq k$ may be replaced with $|x| = k$.
It is well-known (and easy to check) that, for any letter $\alpha$
and any word $w$:
\begin{equation}
\label{eq_binom_aa}
\binom{w}{\alpha\alpha} = \binom{w}{\alpha}\left(\binom{w}{\alpha}-1\right) = |w|_\alpha (|w|_\alpha-1).
\end{equation}

It is also well-known (see, for instance, \cite{Lejeune_Leroy_Rigo2020JCTA}) that, 
for any $k \geq 1$, relation $\sim_k$ is a \textit{congruence},
that is, for any words $u$, $v$, $x$ and $y$, $u \sim_k v$
$\Leftrightarrow$ $xuy \sim_k xvy$. 
For $\sim_2$, the proof of this fact can be illustrated (or done) using Figure~\ref{figure_congruence}.

\input{congruence.tex}

Actually, this figure illustrates the next equation (for $x$, $u$, $y$ words 
and $a$ and $b$ distinct letters): this provides an example of use of the geometrical interpretation provided earlier.

\begin{equation}
\binom{xuy}{ab} = 
\binom{u}{ab}+\binom{x}{ab}+\binom{y}{ab}+
|x|_a|u|_b+|x|_a|y|_b+|u|_a|y|_b \label{eq_binom_3_words}
\end{equation}

Using this equation (and the fact that $|xuy|_\alpha = |x|_\alpha + |u|_\alpha + |v|_\alpha$), 
considering any word $v$, we can verify that 
$u \sim_2 v$ (that is, $|u|_a = |v|_a$, $|u|_b = |v|_b$ and 
$\binom{u}{ab} = \binom{v}{ab}$) if and only if 
$xuy \sim_2 xvy$ (that is, $|xuy|_a = |xvy|_a$, $|xuy|_b = |xvy|_b$ and 
$\binom{xuy}{ab} = \binom{xvy}{ab}$).

Equation~\ref{eq_binom_3_words} can be seen as an extension of the next formula for the concatenation of two words $u$ and $v$.

\begin{equation}
\binom{uv}{ab} = 
\binom{u}{ab}+\binom{v}{ab}+
|u|_a|v|_b \label{eq_binom_2_words}
\end{equation}

\subsection{\label{subsec_storing}Storing and computing the \texorpdfstring{$2$}{2}-binomial equivalence}

The literature contains several works studying characterization of $2$-binomially equivalent words. In particular, A.~{\cerny} \cite{Cerny2009JALC} showed that two words over an arbitrary alphabet are $2$-binomially equivalent if and only if they have the same precedence matrix (see below for the definition). In the binary case, it is also known (see for instance \cite{Fosse_Richomme2004IPL})
that two words are $2$-binomially equivalent
if and only if they are Parikh-equivalent, that is, if they have the same Parikh matrix (see also below for the definition). 
Actually these results are direct consequences of the 
definition of precedence and Parikh matrices and Formulas~\eqref{eq_relation_base1} and \eqref{eq_binom_aa}.
Indeed precedence and Parikh matrices are essentially data structures to store binomial coefficients of some subwords 
of a given a word $w$: in the previous cases, 
they store sufficient information to recover all numbers of occurrences as subwords of words of length $1$ or $2$. 
These matrices are interesting since 
they can be easily computed from the precedence and Parikh matrices of the letters of $w$. We explain and complement this vision.

In \cite{Cerny2009JALC}, 
A.~{\cerny} introduced some arrays that he called \textit{precedence  matrices} or \textit{$p$-matrices}. For a word $w$ over a $k$-letter alphabet 
$A = \{a_1, \ldots, a_k\}$, 
the $p$-matrix associated with the word $w$, denoted $\pmat(w)$, is the $k\times k$ array defined by:
$$\pmat_{ij}(w) = \left\{
\begin{tabular}{l}
$|w|_{a_i}$ if $i = j$\\
$\displaystyle \binom{w}{a_ia_j}$ if $i \neq j$\\
\end{tabular}
\right.
$$

For instance, in the binary case,
$\displaystyle \pmat(w) :=	
	\left[\begin{array}{cc}
			\binom{w}{a}&\binom{w}{ab}\\[.1cm]
			\binom{w}{ba}&\binom{w}{b}
			\end{array}
	\right]$=
	$\displaystyle 
	\left[\begin{array}{cc}
			|w|_{a}&\binom{w}{ab}\\[.1cm]
			\binom{w}{ba}&|w|_{b}
			\end{array}
	\right]$

Actually A.~{\cerny} defined the precedence matrix of a word $w$
from the precedence matrices of letters
and a particular product of $k\times k$
that he denoted $\circ$:
for two $k\times k$ arrays $A$ and $B$, 
$$(A\circ B)_{i, j} = \left\{
\begin{tabular}{l}
$A_{i,i}+B_{i,i}$ if $i = j$\\
$A_{i,j}+B_{i,j}+A_{i,i}B_{j,j}$ if $i \neq j$\\
\end{tabular}
\right.
$$

He observed that, by a simple induction, one can prove:

\begin{lemma}[\cite{Cerny2009JALC}]
\label{L_p_matrix}
For any word $w = w_1 \cdots w_n$ ($w_i \in A$), 
$$\pmat(w) = \pmat(w_1) \circ \cdots \circ \pmat(w_n).$$
\end{lemma}

Let us observe that using Equation~\ref{eq_relation_base1}, one can consider a variant of this approach replacing the values in the bottom triangle of $\pmat(w)$ by $0$. Let define :
$$\pmat_{ij}'(w) = \left\{
\begin{tabular}{l}
$0$ if $i > j$\\
$|w|_{a_i}$ if $i = j$\\
$\displaystyle \binom{w}{a_ia_j}$ if $i < j$\\
\end{tabular}
\right.
$$
In the binary case,
$\displaystyle \pmat'(w) :=	
	\left[\begin{array}{cc}
			\binom{w}{a}&\binom{w}{ab}\\[.1cm]
			0&\binom{w}{b}
			\end{array}
	\right]$=
	$\displaystyle 
	\left[\begin{array}{cc}
			|w|_{a}&\binom{w}{ab}\\[.1cm]
			0&|w|_{a}
			\end{array}
	\right]$.
For arbitrary alphabets, we let readers verify that Lemma~\ref{L_p_matrix} is still valid replacing $\pmat$ with $\pmat'$. 
In the binary case (with an immediate extension to arbitrary alphabet), it may be as much simple to store information in the vector $(|w|_{a}, |w|_{b}, \binom{w}{ab})$ (adapting the operation $\circ$ for computation from vectors associated with letters).

Given a word $w$ over a $k$-letter alphabet $A = \{a_1, \ldots, a_k\}$,
A.~Mateescu et al. defined in 2001 \cite{Mateescu_Salomaa_Salomaa_Yu2001TIA},
the \textit{Parikh matrix} of a word $w$
as the $(k+1)\times(k+1)$ array defined as follows:
for $1 \leq i \leq k+1$, $1\leq j \leq k+1$,
$$\parikhMat_{ij}(w) = \left\{
\begin{tabular}{l}
$1$ if $i = j$\\
$0$ if $i > j$\\
$\displaystyle \binom{w}{a_i\cdots a_j}$ if $i < j$\\
\end{tabular}
\right.
$$
In the binary case $\parikhMat(w)$ := 
	$\displaystyle 
	\left[\begin{array}{ccc} 1 & \binom{w}{a}&\binom{w}{ab}\\[.1cm]
			0& 1 & \binom{w}{b}\\[.1cm]
			0 & 0 & 1
			\end{array}
	\right]$ = 
	$\displaystyle 
	\left[\begin{array}{ccc} 1 & |u|_{a}&\binom{u}{ab}\\[.1cm]
			0& 1 & |u|_{b}\\[.1cm]
			0 &0 & 1
			\end{array}
	\right]$.
This matrix in the binary case already appears in an article \cite{Prodinger1979DM} written by H.~Prodinger in 1979 more than 20 years before the introduction of the terminology ``Parikh matrix". In this paper, he proves in the binary case the following lemma 
in which the matrix product is the usual one.

\begin{lemma}[\cite{Mateescu_Salomaa_Salomaa_Yu2001TIA} and, for the binary case, \cite{Prodinger1979DM}]
\label{L_parikh_matrix}
For any word $w = w_1 \cdots w_n$ ($w_i \in A$), 
$$\parikhMat(w) = \parikhMat(w_1) \cdots \parikhMat(w_n).$$
\end{lemma}

To end this section, we summarize the two results we have mentioned at the beginning of this section and recall also a result linked with Section~\ref{subsec_sum_positions}.

\begin{theorem}[\cite{Cerny2009JALC}]
\label{T_equiv_precedence_matrix}
Given two words $u$ and $v$ over an alphabet $A$, $u \sim_2 v$ if and only if they have the same precedence matrix.
\end{theorem}

\begin{theorem}[\cite{Fosse_Richomme2004IPL}]
\label{T_equiv_parikh_matrix_and_sum_positions}
Given two words $u$ and $v$ over $\{a,b\}$, 
$u \sim_2 v$ if and only if they have the same Parikh matrix
if and only if $\Psi(u) = \Psi(v)$ and $S_b(u) = S_b(v)$.
\end{theorem}

\subsection{\label{subsec_interpret_char}On a characterization using a rewriting rule}

In the previous section, we have recalled that the $2$-binomial equivalence on binary word can be characterized using precedence matrices by Parikh matrices. 
We consider another known characterization of $\sim_2$ in the binary case defined using a rewriting rule.
In this section, we prove this result using the geometrical interpretation of $\binom{w}{ab}$. 
Since the combinatorial proof in \cite{Fosse_Richomme2004IPL}
is shorter, the aim of the proof is more to show how the geometrical interpretation
can be used 
(Note that the intermediary notion of $\Psi$-decomposition will be re-used in Section~\ref{subsecSturm}).

Given a relation $R$ on a set $E$,
we recall that the \textit{reflexive closure} of $R$
is the relation $\{(x,y) \in E\times E \mid (x, y) \in E \text{~or~} (y, x) \in E\}$.
Also the \textit{transitive closure} of $R$  
is the set of all $(x,y) \in E\times E$ 
for which there exist elements $z_0$, $z_1$, \ldots, $z_k$
such that $x = z_0$, $y = z_k$ and $(z_i, z_{i+1}) \in R$ for each $i$, $0 \leq i < k$.
The reflexive and transitive closure of $R$, usually denoted $R^*$,
is the transitive closure of the reflexive closure of $R$.
Such a relation is an equivalence relation.

For two words $u$ and $v$ over the binary alphabet $\{a, b\}$,
let denote\footnote{In \cite{Fosse_Richomme2004IPL} the relation $\rightarrow$ was denoted $\equiv$. The notation is changed in order to 
take care of the natural orientation of the definition}
$u \rightarrow v$ the fact that $u$ and $v$ can be decomposed 
$u = x ab y ba z$ and $v = x ba y ab z$ for some words $x$, $y$ and $z$. 
We also let denote $v \leftarrow u$ this fact.
Relation $\leftarrow$ is the mirror relation of $\rightarrow$: $u \leftarrow v$ if and only if $v \rightarrow u$.

The reflexive closure of these relations is denoted $\equiv$ ($\equiv = \leftarrow \cup \rightarrow$).
The reflexive and transitive closure of relation $\rightarrow$
is denoted $\equiv^*$. 
For instance, $abababa \rightarrow baaabba \rightarrow baabaab$ which implies that $baaabba \equiv abababa$ and $baabaab \equiv^* abababa$.
The whole class of equivalence of
$abababa$ by $\equiv^*$ is $\{aabbbaa, abababa, baaabba, abbaaab, baabaab\}$. 

\begin{theorem}\cite{Fosse_Richomme2004IPL}
\label{T_equivalence_reecrite_sim2}
For two binary words $u$ and $v$, 
$u \equiv^* v$ if and only if $u \sim_2 v$.
\end{theorem}

\begin{proof}[Proof of the only if part of Theorem~\ref{T_equivalence_reecrite_sim2}]
The first part of our proof lies first on next equation
which is quite immediate.
Figure~\ref{fig_interpretation_baisse_aire} explains 
it, showing graphically that when switching the factor $ab$ in $xaby$ to obtain $xbay$, the area at the left of the line representation of the word decreases by 1.

\begin{equation}
\binom{xaby}{ab}=\binom{xbay}{ab}-1 \label{eq_baisse_1}
\end{equation}

\input{passage_xaby_xbay.tex}

Thus switching simultaneously a factor $ab$ to $ba$ and a (non-overlapping) factor $ba$ to factor $ab$ leads to 
an interpretation of the relation 
$\binom{xabybaz}{ab}=\binom{xbayabz}{ab}$ which is a direct consequence of Equation~\eqref{eq_baisse_1} (see Figure~\ref{fig_passage_xabybaz_xbayabz}).

\input{passage_xabybaz_xbayabz.tex}

Since we also have $|xabybaz|_a = |xbayabz|_a$ and
$|xabybaz|_b = |xbayabz|_b$, it follows:
if $u \rightarrow v$ than $u \sim_2 v$.
Consequently: if $u \equiv^* v$ than $u \sim_2 v$.
\end{proof}

\medskip

The idea for the proof of the converse 
is illustrated by Figure~\ref{fig_2_words} that represents two 2-binomially equivalent words.
When $u  \sim_2 v$, 
it may be observed that sometimes 
the line representation of the word $u$ must be above of the line representation of $v$ and at other times
the line representation of the word $u$ must be below of the line representation of $v$.
This implies that $u$ has two non-overlapping
factors $ab$ and $ba$, and,
replacing these factors with $ba$ and $ab$ 
respectively produces
a word $u'$, that is $u' \equiv u$ ($u \rightarrow u'$ or 
$u' \rightarrow u$), which is closer
to $v$ than $u$ (``closer" meaning that the area delimited by the line representations of $u$ and $v$
is strictly smaller than the area delimited by the line representation of $u'$ and $v$). Since $u' \sim_2 v$, we can repeat the analysis. Thus, by induction, 
we can prove that $u \equiv^* v$.

\input{about_2_binomial_equivalence.tex}

\medskip

In order to prove this more formally, we need to introduce some notions and notation that interpret various elements presented in Figure~\ref{fig_2_words}.

\medskip

Let $u$ and $v$ be two words (not necessarily $2$-binomially equivalent). 
The line representations of the two words may cross or join at several points. 
Considering one of this
point and the corresponding prefixes $p_u$ and $p_v$ of respectively $u$ and $v$,
we may observe that $\Psi(p_u) = \Psi(p_v)$. 
Moreover considering
the words between successive such points leads to a natural
decomposition of $u$ and $v$.
A sequence $(u_i, v_i)_{1 \leq i \leq k}$ 
is called a \textit{$\Psi$-decomposition} 
of $(u, v)$ if $u = u_1 \cdots u_k$,
$v = v_1 \cdots v_k$ with,
\begin{itemize}
\item  for all $i$, $1 \leq i \leq k$, $u_i \neq \varepsilon$ and $v_i \neq \varepsilon$,  
and, 
\item for all $i$, $1 \leq i < k$,
$\Psi(u_i) = \Psi(v_i)$ (or equivalently $\Psi(u_1\cdots u_i) = \Psi(v_1\cdots v_i)$).
\end{itemize}
Observe that when $\Psi(u) = \Psi(v)$, we also 
have $\Psi(u_k) = \Psi(v_k)$.
Such a $\Psi$-decomposition is said \textit{maximal} if there exists no longer
$\Psi$-decomposition of $(u, v)$.
Since elements of a $\Psi$-decomposition are not empty, such a maximal sequence exists and is unique (it may be empty if $\Psi(u) \neq \Psi(v)$):
the corresponding sequence $(u_1 \cdots u_i, v_1 \cdots v_i)_{1 \leq i < k}$ is the sequence of pairs of
nonempty proper prefixes of respectively $u$ and $v$ with same Parikh vector.

To illustrate this notion of $\Psi$-decomposition, in Figure~\ref{fig_2_words},
let $u$ be the word represented by a plain red line
and $v$ be the word represented by a dashed black line.
The maximal $\Psi$-decomposition of $(u, v)$
is $((aabab, babaa),$ 
$(b, b),$ 
$(a, a),$
$(ababbb, bbabba),$
$(a, a),$
$(baa, aab),$
$(babaaaa, aaaaabb),$
$(aab, baa))$.

\begin{lemma}
\label{L_Psi_difference}
Let $u$ and $v$ be two words over $\{a, b\}$
with $\Psi(u) = \Psi(v)$. Let $(u_i, v_i)_{1 \leq i \leq k}$
be a $\Psi$-decomposition of $(u, v)$. We have
$$\binom{u}{ab} - \binom{v}{ab} = \sum_{i = 1}^k \binom{u_i}{ab} - \binom{v_i}{ab}$$
\end{lemma}

\begin{proof}
The result is immediate when $k = 1$. Assume $k = 2$.
By Formula~\ref{eq_binom_2_words} and 
since $|u_1|_a = |v_1|_a$ and $|u_2|_b = |v_2|_b$,
$\displaystyle 
\binom{u_1u_2}{ab} - \binom{v_1v_2}{ab} = 
\left[\binom{u_1}{ab} - \binom{v_1}{ab}\right]
+ \left[\binom{u_2}{ab} - \binom{v_2}{ab}\right]
+ |u_1|_a |u_2|_b - |v_1|_a |v_2|_b =
\left[\binom{u_1}{ab} - \binom{v_1}{ab}\right]
+ \left[\binom{u_2}{ab} - \binom{v_2}{ab}\right]$.
This proves the lemma for $k = 2$.

Let $(u_i, v_i)_{1 \leq i \leq k}$ be a $\Psi$-decomposition of $(u, v)$. 
Let $u_i' = u_i$ and $v_i' = v_i$ for all $i$, $1 \leq i \leq k-2$.
Let $u_{k-1}' = u_{k-1}u_k$ and $v_{k-1}' = v_{k-1}v_k$.
It is a simple observation that 
$(u_i', v_i')_{1 \leq i \leq k-1}$ is a $\Psi$-decomposition of $(u,v)$.
The proof of the lemma ends
by induction.
\end{proof}

A pair of binary words $(u, v)$ with $\Psi(u) = \Psi(v)$
is said to be \textit{$\Psi$-undecomposable}
if its maximal $\Psi$-decomposition is of length $1$
(the definition can be extended to pairs of words with
different Parikh vectors imposing in such a case that the length
of the maximal decomposition is $0$).
Being $\Psi$-undecomposable means that
there exist no pair $(p, p')$
of nonempty proper prefixes of $u$ and $v$ such that $\Psi(p) = \Psi(p')$.
Observe that if $(u_i, v_i)_{1 \leq i \leq k}$ is a maximal $\Psi$-decomposition, then each pair $(u_i, v_i)$ is $\Psi$-undecomposable. 
For nay pair of different words,
the $\Psi$-undecomposability implies that
the first letters of the two elements of the pair
are different (also the last letters are different).

\begin{lemma}
\label{L_Psi_value}
Let $u$ and $v$ be two words over $\{a, b\}$
with $\Psi(au) = \Psi(bv)$. Assume
$(au, bv)$ is $\Psi$-undecomposable. We have:
\begin{enumerate}
\item $|\pref_{i, b}(au)|_a > |\pref_{i, b}(bv)|_a$ for each $i$, $1  \leq i \leq |au|_b$. 
\item $\displaystyle\binom{au}{ab} - \binom{bv}{ab} > 0$.
\end{enumerate}
\end{lemma}

\begin{proof}
Let us prove the first item.
Observe that for $i = 1$, $|\pref_{1, b}(au)|_a \geq 1 > 0 = |\pref_{1, b}(bv)|_a$.
Assume by contradiction that there
exists an integer $i$, $1 \leq i < |au|_b$, 
such that $|\pref_{i+1, b}(au)|_a \leq |\pref_{i+1, b}(bv)|_a$
and consider the smallest such integer: $|\pref_{i, b}(au)|_a > |\pref_{i, b}(bv)|_a$.
We have:
\begin{equation}
|\pref_{i, b}(bv)| < |\pref_{i, b}(au)| \leq |\pref_{i+1, b}(bv)|-1 \label{eqproof1}
\end{equation}

$\begin{array}{cl}
\text{Indeed~} |\pref_{i, b}(bv)| 
	&= i + |\pref_{i, b}(bv)|_a\\
	&< i + |\pref_{i, b}(au)|_a =|\pref_{i, b}(au)|\\
	&\leq i + |\pref_{i+1, b}(au)|_a \leq i + |\pref_{i+1, b}(bv)|_a\\
	&= |\pref_{i+1, b}(bv)| + i - |\pref_{i+1, b}(bv)|_b\\
	&=|\pref_{i+1, b}(bv)| -1
\end{array}$

Set $p_i = \pref_{i, b}(au)$ and $p_i' = \pref_{i, b}(bv)$.
Since the prefix of length $|\pref_{i+1, b}(bv)| -1$ of $bv$
belongs to $p_i'a^*$,
Equation~\eqref{eqproof1} implies that
the prefix of length $|p_i|$ of $bv$
is $p_i' a^{|p_i|-|p_i'|}$.
Observe that
$\Psi(p_i) = \Psi(p_i' a^{|p_i|-|p_i'|})$.
This contradicts the fact that $(au, bv)$ 
is $\Psi$-undecomposable. This ends the proof of the first item.

The second item is a direct consequence of the first item 
and 
of Formula~\eqref{eq_calcul_binom_ab_par_prefixes}.
\end{proof}

Given a pair of words $(u, v)$ with $\Psi(u) = \Psi(v)$,
the space between the line representations of the two words 
is the symmetric 
difference $\leftPart(u) \Delta \leftPart(v)$, or equivalently,
$\rightPart(u) \Delta \rightPart(v)$. Let denote $d(u, v) := \#(\leftPart(u) \Delta \leftPart(v)) =
\#(\rightPart(u) \Delta \rightPart(v))$. 
The function $d$ is a distance. The facts that ``$d(x, x) = 0$ for any word $x$"
and ``$d(x, y) = d(y, x)$ for any words $x$ and $y$" 
follow immediately the definition.
To see that ``$d(x, y)+d(y, z) \geq  d(x, z)$ 
for any words $x$, $y$ and $z$", it is sufficient 
to observe that if $(i, j) \not\in (\leftPart(x) \Delta \leftPart(y)) \cup (\leftPart(y) \Delta \leftPart(z))$
then also $(i, j) \not\in (\leftPart(x) \Delta \leftPart(z))$.

Note that $d$ is defined on arbitrary pairs of words that are not
necessarily $2$-binomially equivalent. For instance
$d(ab, ba) = 1$.

When $u \sim_2 v$, observe that $d(u,v)$ is even.
Indeed $d(u,v) = \#\leftPart(u) + \#\leftPart(v) - 2 \#(\leftPart(u) \cap \leftPart(v))$
and by Lemma~\ref{L_interpretation_binom_ab}, $\#\leftPart(u) =$ $\binom{u}{ab} =$ $\binom{v}{ab} =$ $\#\leftPart(v)$:
$d(u,v) = 2(\#\leftPart(u)-\#(\leftPart(u) \cap \leftPart(v))$.

\begin{lemma}
\label{L_d_dec}
Let $u$ and $v$ be two words over $\{a, b\}$
with $\Psi(u) = \Psi(v)$. 
Let $(u_i, v_i)_{1 \leq i \leq k}$
be a $\Psi$-decomposition of $(u, v)$.

\begin{equation}
\label{eq_somme_d}
d(u, v) = \sum_{i = 1}^k d(u_i, v_i)
\end{equation}
\end{lemma}

\begin{proof}
When $k = 1$, Formula~\eqref{eq_somme_d} is trivial.

Assume $k = 2$.
Let recall that by definition 
$\leftPart(w) = \{ (i, j) \mid 1 \leq j \leq |w|_b, 1 \leq i \leq |\pref_{j, b}(w)|_a\}$ for any word $w$.
Thus $\leftPart(u_1u_2) = 
\leftPart(u_1) \cup R \cup S$
where $R = \{ (i, j) \mid |u_1|_b+1 \leq j \leq |u_1u_2|_b, 
1 \leq i \leq |u_1|_a \}$
and $S = \{ (i, j) \mid |u_1|_b+1 \leq j \leq |u_1u_2|_b,
|u_1|_a+1  \leq i \leq |\pref_{j,b}(u_1u_2)|_a \}$.
Since for $j \geq |u_1|_b+1$,
$|\pref_{j,b}(u_1u_2)|_a = |u_1|_a+|\pref_{j-|u_1|_b,b}(u_2)|_a$,
observe that
$S = \{ (|u_1|_a+i, |u_1|_b+j) \mid
1 \leq j \leq |u_2|_b, 1\leq i \leq |\pref_{j,b}(u_2)|_a\}$.
In other words
$S = \Psi(u_1)+\leftPart(u_2)$.
Since $((u_1,v_1),(u_2,v_2))$ is a $\Psi$-decomposition of $(u,v)$,
we have $\Psi(u_1) = \Psi(v_1)$ which means that
$|u_1|_a=|v_1|_a$ and $|u_1|_b = |v_1|_b$.
It follows that $R = \{ (i, j) \mid |v_1|_b+1 \leq j \leq |v_1v_2|_b, 
1 \leq i \leq |v_1|_a \}$. 
Consequently
$\leftPart(v_1v_2) = \leftPart(v_1) \cup R \cup (\Psi(v_1)+ \leftPart(v_2))$.
We deduce from what precedes that
\begin{equation}
\leftPart(u_1u_2)\Delta\leftPart(v_1v_2) =
\leftPart(u_1)\Delta\leftPart(v_1)
\cup
(\Psi(u_1)+\leftPart(u_2))\Delta(\Psi(v_1)+ \leftPart(v_2)) \label{eqproof2}
\end{equation}

Observe that $(i,j) \in (\Psi(u_1)+\leftPart(u_2))\Delta(\Psi(v_1)+ \leftPart(v_2))$
if and only if 
$(i-|u_1|_a,j-|u_1|_b) \in \leftPart(u_2)\Delta\leftPart(v_2)$.
Hence $\#(\Psi(u_1)+\leftPart(u_2))\Delta(\Psi(v_1)+ \leftPart(v_2))
= \#(\leftPart(u_2)\Delta\leftPart(v_2)) = d(u_2,v_2)$.
Since $d(u, v) = \#(\leftPart(u_1u_2)\Delta\leftPart(v_1v_2))$
and $d(u_1,v_1) = \#(\leftPart(u_1)\Delta\leftPart(v_1))$,
we get $d(u,v) = d(u_1, v_1)+d(u_2, v_2)$.

The proof of Formula~\ref{eq_somme_d} ends by induction.
\end{proof}

\begin{lemma}
\label{L_dist_Psi_undec}
Let $u$ and $v$ be two words over $\{a, b\}$
with $\Psi(u) = \Psi(v)$ and $(u, v)$ $\Psi$-undecomposable.
$$d(u,v) = \left|\binom{u}{ab}-\binom{u}{ba}\right|$$
\end{lemma}

\begin{proof}
The fact that $(u,v)$ is $\Psi$-decomposable
implies that words $u$ and $v$ begin with different letters.
Assume that $u$ begins with $a$ and $v$ begins with $b$.
By Lemma~\ref{L_Psi_value}, for each $j$, $1\leq j \leq |u|_b$,
$|\pref_{j,b}(u)|_a > |\pref_{j,b}(bv)|_a$.
Using the fact that $\Psi(u) = \Psi(v)$, that is, $|u|_a=|v|_a$ 
and $|u|_b=|v|_b$, it follows that
$$\leftPart(u)\Delta\leftPart(v) = \leftPart(u)\setminus\leftPart(v).$$
Indeed:
\begin{itemize}
\item $\leftPart(u) = \{(i, j) \mid 1 \leq j \leq |u|_b, 1 \leq i \leq |\pref_{j,b}(u)|_a\}$
\item $\leftPart(v) = \{(i, j) \mid 1 \leq j \leq |v|_b, 1 \leq i \leq |\pref_{j,b}(v)|_a\} \subseteq \leftPart(u)$
\item $\leftPart(u)\Delta\leftPart(v) = \{(i, j) \mid 1 \leq j \leq |u|_b, 
|\pref_{j,b}(v)|_a+1 \leq i \leq |\pref_{j,b}(u)|_a\}$.
\end{itemize}
Thus $d(u, v) = \#\leftPart(u)-\#\leftPart(v)$.
Lemma~\ref{L_interpretation_binom_ab} implies
$$d(u,v) = \binom{u}{ab}-\binom{v}{ab}.$$

Similarly when $u$ begins with $b$ and $v$ begins with $a$, 
$$d(u,v) = \binom{v}{ab}-\binom{u}{ab}.$$
\end{proof}

\begin{proof}[Proof of the if part of Theorem~\ref{T_equivalence_reecrite_sim2}] 
Let $u$ and $v$ be two different words with $u \sim_2 v$.
Let $(u_i, v_i)_{1 \leq i \leq k}$
be the maximal $\Psi$-decomposition of $(u, v)$.

Lemma~\ref{L_Psi_difference} and the fact that $\binom{u}{ab} = \binom{v}{ab}$ imply that 
$\sum_{i = 1}^k \binom{u_i}{ab} - \binom{v_i}{ab} = 0$.
Hence there exists $\alpha$ such that 
$\binom{u_\alpha}{ab} - \binom{v_\alpha}{ab} > 0$
and there exists $\beta$ such that 
$\binom{u_\beta}{ab} - \binom{v_\beta}{ab} < 0$

By Item~2 of Lemma~\ref{L_Psi_value}, the first letter of $u_\alpha$ and $v_\beta$ is $a$ 
and the first letter of $u_\beta$ and $v_\alpha$ is $b$. The fact that $\Psi(u_\alpha) = \Psi(v_\alpha)$ and $\Psi(u_\beta) = \Psi(v_\beta)$ implies that $a$ and $b$ occur both in $u_\alpha$ and $u_\beta$. Hence $ab$ is a factor of $u_\alpha$
and $ba$ is a factor of $u_\beta$.

Let $x$, $y$, $z$, $t$ be words such that $u_\alpha = xaby$
and $u_\beta = zbat$. 
Consider the sequence $(u_\gamma')_{1 \leq i \leq k}$
defined by $u_\alpha' = xbay$, $u_\beta' = zabt$ and
$u_\gamma' = u_\gamma$ for $\gamma \not\in \{\alpha, \beta\}$.
Let $u' = u_1'\cdots u_k'$. 
It is immediate that: $|u'|_a = |u|_a$, $|u'|_b = |u|_b$.
Observe also that $(u'_i, v_i)_{1 \leq i\leq k}$ is a $\Psi$-decomposition of $(u',v)$.
From Formula~\eqref{eq_baisse_1},
$\displaystyle \binom{u_\alpha'}{ab} = \binom{u_\alpha}{ab}-1$ and
$\displaystyle \binom{u_\beta'}{ab} = \binom{u_\beta}{ab}+1$.
Observe that $\displaystyle \binom{u_\alpha'}{ab}+\binom{u_\beta'}{ab} = \binom{u_\alpha}{ab}+\binom{u_\beta}{ab}$ and
$\displaystyle 
\left|\binom{u_\alpha}{ab}-\binom{v_\alpha}{ab}\right|+
\left|\binom{u_\beta}{ab}-\binom{v_\alpha}{ab}\right| = 
\left|\binom{u_\alpha}{ab}-\binom{v_\alpha}{ab}\right|+
\left|\binom{u_\beta}{ab}-\binom{v_\beta}{ab}\right|-2$.
Thus Lemma~\ref{L_Psi_difference} implies that
$\binom{u'}{ab} = \binom{u}{ab}$ (hence $u' \sim_2 u$), and Lemma~\ref{L_d_dec} implies
that $d(u', v) = d(u, v)-2$.

From what precedes by induction we can find a sequence of words 
$(U_i)_{0 \leq i \leq m}$ with $U_0= u$ and $U_m = v$ 
such that
for all $i$, $0 \leq i \leq m$, $U_i \sim_2 u$ and $U_i \equiv u$, and, for all $i$, $0 \leq i \leq m$, $d(U_i,v) = d(u, v)-2i$ (to end the induction, it is useful to know that $d(u, v)$ is even due to the fact that $u \sim_2 v$). In particular, $u \equiv^* v$.
\end{proof}

\begin{remark}
\label{rem_length_derivation}
In the previous proof $m = d(u, v)/2$.
The derivation sequence $(U_i)_{0 \leq i \leq m}$ from $u$ to $v$ is of minimal length. 
Indeed from Formula~\eqref{eq_baisse_1} it follows that, for any word $x$ and $y$ with $x \rightarrow y$ or more generally $x \equiv y$ and for any word $z$, $d(x, z) - d(y, z) \in \{-2, 0, +2\}$.
\end{remark}

\section{\label{secPartitions}On the structure of a \texorpdfstring{$2$}{2}-binomial equivalence class}

This section concerns only binary words.

\subsection{\label{subsec_first_lattice}A first lattice structure}
Let $[w]_{\sim_2}$ denote the equivalence class of the word $w$ for
the $2$-binomial equivalence $\sim_2$.
Let $G_2(w) = ([w]_{\sim_2}, \rightarrow)$ be the graph of the relation $\rightarrow$ restricted to the set $[w]_{\sim_2}$.
Observe that if $u \rightarrow v$ (for $u$, $v \in \{a, b\}^*$), 
then $u$ has a prefix $pa$ and $v$ has a prefix $pb$ for some word $p$.
Hence:

\begin{fact}\label{fact_ordre}
Given two words $u$ and $v$ with $u \rightarrow v$ or more generally $u \rightarrow^* v$, we have $u <_{lex} v$.
\end{fact}

Since moreover the rewriting rule $\rightarrow$ preserves the length of the word
(if $u \rightarrow v$, then $|u| = |v|$), it follows
that any path in $G_2(w)$ is finite.
In other words, the graph $G_2(w)$ is a directed acyclic graph, or, 
the relation $\rightarrow$ defines a partially order
on the set $[w]_{\sim_2}$. Actually the structure of $G_2(w)$ is stronger.

\begin{theorem}
\label{T_lattice}
Given any word $w$ over $A= \{a, b\}$,
the graph $G_2(w) = ([w]_{\sim_2}, \rightarrow)$ is a lattice.
Its least element is the unique word in 
$[w]_{\sim_2} \cap A^*\setminus A^*baA^*abA^*$ 
and its greatest element is the unique word in 
$[w]_{\sim_2} \cap A^*\setminus A^*abA^*baA^*$.
\end{theorem}

Figure~\ref{fig_lattice_example} illustrates this result.

\input{exemple_graphes_classe.tex}

The following results are useful for the proof of Theorem~\ref{T_lattice}.
The second and third ones explain how can be reconstructed the particular words 
appearing in Theorem~\ref{T_lattice} from the information $|w|_a$, $|w|_b$ and $\binom{w}{ab}$ characterizing the class $[w]_{\sim_2}$.

\begin{lemma}
\label{L_lang_sans_ab...ba}
Let $A = \{a, b\}^*$. We have:
$A^* \setminus A^*abA^*baA^* = b^*a^*b^* \cup b^*a^+ba^+b^*$.
\end{lemma}

\begin{proof}
Let $w$ be a word in $A^* \setminus A^*abA^*baA^*$. 
If it contains no factor in the form $ab^ia$ with $i > 0$, then one can check that
$w \in b^*a^*b^*$. 
If $w$ contains a factor in the form $ab^ia$ with $i >0$, then
$i = 1$ and $w$ contains only one occurrence of $aba$.
It follows that $w$ belongs to $b^*a^+ba^+b^*$. 
Conversely $b^*a^*b^* \cup b^*a^+ba^+b^* \subseteq A^* \setminus A^*abA^*baA^*$.
\end{proof}

Exchanging the role of letters $a$ and $b$, Lemma~\ref{L_lang_sans_ab...ba} also states that $A^* \setminus A^*baA^*abA^* = a^*b^*a^* \cup a^*b^+ab^+a^*$.

\begin{lemma}
\label{L_reconstruct_final}
A word $w$ in $A^* \setminus A^*abA^*baA^*$
is uniquely determined by the values of $|w|_a$, $|w|_b$ and
$\binom{w}{ab}$.
More precisely:
\begin{itemize}
\item if $|w|_a = 0$, $w = b^{|w|_b}$,
\item if $\binom{w}{ab} = |w|_a|w|_b$, $w = a^{|w|_a}b^{|w|_b}$,  and
\item 
if $|w|_a \neq 0$ and $\binom{w}{ab} \neq |w|_a|w|_b$,
$w = b^{|w|_b-i-1}a^jba^{|w|_a-j}b^i$ where $i$ and $j$
are the unique integers such that $\binom{w}{ab} = i |w|_a + j$ 
with $0 \leq j < |w|_a$.
\end{itemize}
\end{lemma}

\begin{remark}
Since $\binom{w}{ab} \leq |w|_a|w|_b$, 
we can verify that, in the last part of the previous lemma,
$0 \leq i \leq |w|_b$. In the extreme case where $i = |w|_b$, 
we have $j = 0$ and $w = a^{|w|_a}b^{|w|_b}$. In other cases,
$|w|_b-i-1 \geq 0$.
\end{remark}

\begin{proof}[Proof of Lemma~\ref{L_reconstruct_final}]
Let $w$ in $A^* \setminus A^*abA^*baA^*$.
When $|w|_a = 0$, the lemma holds.

From now on, let us assume that $|w|_a \neq 0$.
By Lemma~\ref{L_lang_sans_ab...ba}, $w \in b^*a^*b^* \cup b^*a^+ba^+b^*$.

Assume first that $w = b^ia^{|w|_a}b^k$ 
for some integers $i$ and $k$. 
The relation $\binom{w}{ab}= |w|_a k$ 
implies that $k = \binom{w}{ab}/|w|_a$. 
It is straightforward that $i = |w|_b-k$. 
When $k < |w|_b$ ($\binom{w}{ab} \neq |w|_a|w|_b$),
we have $i \geq 1$.
The lemma is verified in this case
since $\binom{w}{ab} = k |w|_a + j$ with $j = 0$.

Assume now that $w = b^ia^jba^kb^\ell$ for some integers $i \geq 1$, $j \geq 1$, $k$ and $\ell$.
We have $\binom{w}{ab} = \ell |w|_a + j$ with $0 \leq j < |w|_a$,
$k = |w|_a-j$ and $i = |w|_b-1-\ell$.
The lemma holds in this last case.
\end{proof}

Exchanging the role of letters $a$ and $b$, 
Lemma~\ref{L_reconstruct_final} provides a way to reconstruct a word in $A^* \setminus A^*baA^*abA^*$ from the values of $|w|_a$, $|w|_b$ and 
$\binom{w}{ba}$. Next lemma does the same but using $\binom{w}{ab}$
instead of $\binom{w}{ba}$. Its proof is similar to the proof of Lemma~\ref{L_reconstruct_init} and it is let to readers.

\begin{lemma}
\label{L_reconstruct_init}
A word $w$ in $A^* \setminus A^*baA^*abA^*$
is uniquely determined by the values of $|w|_a$, $|w|_b$ and
$\binom{w}{ab}$.
More precisely
\begin{itemize}
\item 
if $|w|_b = 0$, $w = a^{|w|_a}$,
\item if $\binom{w}{ab} = |w|_a|w|_b$, $w = a^{|w|_a}b^{|w|_b}$, 
and
\item if $|w|_b \neq 0$ and $\binom{w}{ab} \neq |w|_a|w|_b$, $w = a^{i}b^{|w|_b-j}ab^{j}a^{|w|_a-i-1}$ where $i$ and $j$
are the unique integers such that $\binom{w}{ab} = i |w|_b + j$ 
with $0 \leq j < |w|_b$.

\end{itemize}
\end{lemma}

\begin{proof}[Proof of Theorem~\ref{T_lattice}]
Let us construct a sequence of words: $w_1= w$ and, for any integer $i$ with $w_i$ defined, if $w_i$
contains a factor in $ab\{a, b\}^*ba$, let $w_{i+1}$ be any word such that 
$w_i \rightarrow w_{i+1}$. 
We continue the construction as long as possible.
From Fact~\ref{fact_ordre}, $w_i <_{lex} w_{i+1}$. Since the length of all words $w_i$ are the same, the constructed sequence is finite. With $n$ the length of this sequence, this means that $w_n$ contains no factor in $ab\{a, b\}^*ba$: $w_n \in A^* \setminus A^*abA^*baA^*$ and by construction $w \rightarrow^* w_n$. Set $w_{max} = w_n$. While constructing 
the sequence $(w_i)_{1\leq i \leq n}$, we may have made some choice. Note that Lemma~\ref{L_reconstruct_final} implies that, whatever is this choice, the last element of the sequence is always $w_{max}$.

Taking any word $v$ in $[w]_{\sim_2}$. We can similarly 
find a word $v' \in A^* \setminus A^*abA^*baA^*$ with $v \rightarrow^* v'$. Once again Lemma~\ref{L_reconstruct_final} implies that $v' = w_{max}$. This proves the theorem for the greatest element.

The proof for the least element is similar considering a sequence of words using relation $\leftarrow$ and Lemma~\ref{L_reconstruct_init}
\end{proof}

From now on, let $\init(w)$ and $\final(w)$ denote
respectively the least and the greatest elements of $[w]_{\sim_2}$.

Let us illustrate the constructions provided in
Lemma~\ref{L_reconstruct_final} and \ref{L_reconstruct_init} 
(see also Figure~\ref{fig_exemple_init_and_final} where $i$ and $j$ are as in \ref{L_reconstruct_init} and 
Lemma~\ref{L_reconstruct_final} respectively).
Consider the word $w = aabaabbababba$: $|w|_a = 7$, $|w|_b = 6$, $\binom{w}{ab} = 27$. Since $27 = 4 |w|_b + 3$, by Lemma~\ref{L_reconstruct_init}, $\init(w) = a^4b^3ab^3a^2$.
Since $27 = 3 |w|_a + 6$, by Lemma~\ref{L_reconstruct_final},  $\final(w) = b^2a^6bab^3$.

\input{exemple_init_final.tex}

It may be observed that previous results have the following corollary.

\begin{corollary}
For two binary words $u$ and $v$, $u \sim_2 v$ if and only if $\init(u) = \init(v)$ if and only $\final(u) = \final(v)$.
\end{corollary}

\begin{remark}
Theorem~\ref{T_lattice} implies that: 
for any word $w$, $\#[w]_{\sim_2} = 1$ if and only if $\init(w) = \final(w)$. This provides a new characterization of classes that are singletons. Theorem~3 in  \cite{Mateescu_Salomaa2004IJFCS} states that, 
for any $w \in \Sigma^*$ with $\Sigma = \{a, b\}$,
$\#[w]_{\sim_2} = 1$ if and only if
$w \not\in \Sigma^*ab\Sigma^*ba\Sigma^* \cup \Sigma^*ba\Sigma^*ab\Sigma^*$
if and only if $w \in b^*a^* + b^*ab^* + b^*aba^* + a^*b^* + a^*ba^* + a^*bab^*$. Figure~\ref{fig_singletons} shows the various shapes of these words.
It may be noted that if $|w|_a \geq 2$ and $|w|_b \geq 2$ (that is $w \not\in a^*ba^*\cup b^*ab^*$, then $\#[w]_{\sim_2} = 1$ if and only if 
$\binom{w}{ab} \in \{0, 1, |w|_a|w|_b-1, |w|_a|w|_b\}$.
\end{remark}

\input{singletons.tex}

\begin{remark}
Theorem~\ref{T_lattice} also implies that,
if $u \sim_2 v$,
there exists a word $x$ such that $u \rightarrow^* x$ and
$v \rightarrow^* x$. Hence $u \equiv^* x \equiv^* v$. This provides 
a new proof of the if part of Theorem~\ref{T_equivalence_reecrite_sim2}.
\end{remark}

\begin{remark}\label{rem_bibli_init_final} (Bibliographic notes)
The fact that final words are representants
of the $2$-binomial class of binary word is used in \cite[Lemma 4]{Rigo_Salimov2015TCS} to show that
$$\#(A^n/\!\!\sim_2) = \frac{n^3+5n+6}{6}.$$

In \cite{Lejeune_Rigo_Rosenfeld2020IJAC}, 
the previous result of \cite{Rigo_Salimov2015TCS} 
is used (see \cite[Remark 4]{Lejeune_Rigo_Rosenfeld2020IJAC}) 
to determine the number of elements
of the set $LL(\sim_2, \{1, 2\})$ defined as the set $\{ w \in \{1, 2\}^* \mid \forall u \in [w]_{\sim_2} : w \leq_{lex} u\}$: in other 
words $LL(\sim_2, \{1, 2\})$ is the set of initial words over $\{1, 2\}$ that is the set of words that do not contain any factor $21y12$.
\end{remark}

\begin{remark}
It is a consequence of Remark~\ref{rem_length_derivation} that
a minimal path from $\init(w)$ to $\final(w)$ in $G_2(w)$ is 
of length $d(\init(w),\final(w))/2$ (also in the proof of the if part of Theorem~\ref{T_equivalence_reecrite_sim2} taking $u = \init(w)$ and $v = \final(w)$, we can verify that $\alpha < \beta$ due to the fact that $v$ is the greatest element of $[w]_{\sim_2}$ for the lexicographical order: this implies that $U_i \rightarrow U_{i+1}$ for all $i$). 
\end{remark}

Next section will provide an information about the length of maximal paths from $\init(w)$ to $\final(w)$ in $G_2(w)$. See Remark~\ref{rem_longest_paths}.

\subsection{\label{subsec_bij_lattice_integers}Bijections with sets of partitions of integers}

Under the scope of the study of Parikh matrices, in \cite{Teh_Kwa2015TCS} (see also \cite{Mathew_Thomas_Bera_Subramanian2019AAM,Teh2015IJFCS}), 
W.~C.~Teh and K.~H.~Kwa  state a connection between 
 partitions of integers and
 binomial coefficients of some special words that they called \textit{core words} (a first link between Parikh matrices (and so  the binomial coefficient they contains) and partitions 
 can also be found in \cite[Section 4]{Mateescu_Salomaa2004IJFCS}). 
The \textit{core word} $c$ of a word $w \in \{a,b \}^*$ is the unique word in $\varepsilon \cup a\{a,b \}^*b$ such 
that $w \in b^*ca^*$.
In this section, we consider this link between partitions and binomial coefficients
in a slightly different way 
restricting our attention to the $2$-binomial equivalence class of a fixed arbitrary word, that is, 
fixing the number of occurrences of $a$ and $b$ and the number of occurrences of $ab$ as a subword.

Let recall that a \textit{partition}
into k \textit{parts} of a nonnegative integer $n$
is given by a sequence $(\lambda_1, \ldots, \lambda_k)$ 
of non-increasing positive integers such that $n = \sum_{i = 1}^k \lambda_i$ (see \cite{Brylawski1973DM,Greene_Kleitman1986EJC} for instance). 
Note that $\lambda_1$ is the largest part of the partition and
possibly $\lambda_i =  0$ for the last parts.
Let $w$ be a word over $\{a, b\}$ with $|w|_b = k$
and let $p_1b$, \ldots, $p_kb$ be the prefixes of 
$w$ ending with $b$ enumerated by increasing length ($p_k b = \pref_{k,b}(w)$).
Formula~\ref{eq_calcul_binom_ab_par_prefixes} shows
that the sequence $(|p_i|_a)_{i = |w|_b, \ldots, 1}$
forms a partition of the integer $\binom{w}{ab}$.
We denote it $\partition_p(w)$:
$\partition_p(w) = (|p_{|w|_b}|_a, \ldots, |p_1|_a)$.
Observe that the number of parts of the partition is 
$|w|_b$.
The maximal part of the partition is bounded by $|w|_a$.
Note that given any word $v \sim_2 w$, 
we obtain similarly another partition of $\binom{w}{ab} = \binom{v}{ab}$ with $|v|_b = |w|_b$ parts and maximal part bounded by 
$|v|_a = |w|_a$.

Conversely let us consider a partition
$(\lambda_1, \cdots, \lambda_k)$ 
of a positive integer $n$ 
with
$\max_{i= 1,\ldots, n}\lambda_i \leq \ell$ for some integers
$k$ and $\ell$. 
Let $\lambda_{k+1} = 0$
and let 
$w = (\prod_{i = k}^1 a^{\lambda_{i}-\lambda_{i+1}}b)a^{\ell-\lambda_1}$
(that is, $w
= a^{\lambda_k}b
a^{\lambda_{k-1}-\lambda_k}b\cdots
ba^{\lambda_2-\lambda_3}
\cdots ba^{\lambda_{1}-\lambda_{2}}ba^{\ell-\lambda_1}$).
The construction of this word ensures that $\partition_p(w) = (\lambda_1, \cdots, \lambda_k)$, $|w|_a = \ell$ and $|w|_b = k$.
Moreover using Formula~\eqref{eq_calcul_binom_ab_par_prefixes}, 
one can verify that 
$\binom{w}{ab} = \sum_{i=1}^k \lambda_i = n$. 

The two previous paragraphs prove the next result.

\begin{lemma}
\label{L_bijection_partition}
Given a word $w$, 
there is a bijection between the set of elements of $[w]_{\sim_2}$
and the set of partitions of $\binom{w}{ab}$ into $|w|_b$
parts and maximal part bounded by $|w|_a$.
\end{lemma}

In other words, given three integers $k$, $\ell$ and $m$
there is a bijection between the two following sets:
\begin{itemize}
\item the set of words $w$ such that
$|w|_a = k$, $|w|_b = \ell$ and $\binom{w}{ab} = m$;
\item the set of partitions of the integer $m$
with $\ell$ parts all bounded by $k$.
\end{itemize}

To illustrate what precedes, let us consider the word
$w = aabaabbababba$ (see Figure~\ref{fig_aabaabbababba}).
The partition of $27$ associated with this word
is the sequence $\partition_p(w) = (6, 6, 5, 4, 4, 2)$.
The partition associated with $\init(w)$ and $\final(w)$ (see Figure~\ref{fig_exemple_init_and_final}) are respectively
$\partition_p(\init(w)) = (5, 5, 5, 4, 4, 4)$ and $\partition_p(\final(w)) = (7, 7, 7, 6, 0, 0)$. 
It is usual to represent partitions by a Ferrer diagrams (in which the parts of value 0 are not represented).
This diagram can be visualized in our geometrical representations
applying an anticlockwise quarter-turn rotation.

It is well-known that there is a duality between
the set of partitions of an integer with $k$ parts bounded by $\ell$ and the set of partitions of the same integer 
with $\ell$ parts bounded by $k$.
We can find this dual partition
using Formula~\ref{eq_calcul_binom_ab_par_sufixes}.
Let define $\partition_s(w) = (|s_1|_b, \ldots, |s_{|w|_a}|_b)$
where $as_1$, \ldots, $as_{|w|_a}$ are the suffixes of 
$w$ beginning with $a$ enumerated by decreasing length 
(for $1 \leq i \leq |w|_a$, $as_i = \suff_{|w|_a-i+1, a}(w))$.
In a way similar to the prof of Lemma~\ref{L_bijection_partition}, 
we can prove:
\begin{lemma}
\label{L_bijection_partition_dual}
Given a word $w$, 
there is a bijection between the set of elements of $[w]_{\sim_2}$
and the set of partitions of $\binom{w}{ab}$ into at most $|w|_a$
parts and maximal part bounded by $|w|_b$.
\end{lemma}

With $w = aabaabbababba$  (see Figures~\ref{fig_aabaabbababba}
and \ref{fig_exemple_init_and_final}), 
the partitions of $27$ associated with $w$,
$\init(w)$ and $\final(w)$
by the bijection behind Lemma~\ref{L_bijection_partition_dual} 
are respectively the sequences $\partition_s(w) = (6, 6, 5, 5, 3, 2)$,
$\partition_s(\init(w)) = (6, 6, 6, 6, 3, 0, 0)$ and 
$\partition_s(\final(w)) = (4, 4, 4, 4, 4, 4, 3)$. 
Their Ferrer's diagrams can be visualized by symmetry with the top border of the rectangle.

Let us note the following corollary of any of the two previous lemmas.

\begin{corollary}
\label{cor_number_of_elements}
Given a word $w$ with $\binom{w}{ab} \leq \min(|w|_a, |w|_b)$,
$$\#[w]_{\sim_2} = \text{number of partitions of the integer} \binom{w}{ab}.$$
\end{corollary}

In \cite{Brylawski1973DM}, Brylawski proves 
that the set $L_n$ of partitions in $n$ parts
of a nonnegative integer $n$ is a lattice
for the \textit{dominance ordering} defined as follows. 
For $\lambda = (\lambda_1, \ldots, \lambda_n)$ and 
$\mu = (\mu_1, \ldots, \mu_n)$ two partitions of the integer $n$,
$\lambda \geq \mu$ if for each $j$, $1 \leq j \leq n$, 
$\sum_{i = 1}^j \lambda_i \geq \sum_{i = 1}^j \mu_i$.
Traditionally in order theory, an element $\lambda$ covers
another $\mu$ if there exist no elements $\nu$ with
$\lambda > \nu > \mu$.
Brylawski characterizes the covers in $L_n$.
Let $\lambda \succ \mu$ denote
the fact that $\lambda$ covers $\mu$.

\begin{lemma}{\cite[Prop. 2.3]{Brylawski1973DM}}
\label{L_covers_Ln}
In $L_n$, $\lambda = (\lambda_1, \ldots, \lambda_n)$ covers $\mu = (\mu_1, \ldots, \mu_n)$
if and only if there exist $j$ and $k$ such that 
$\lambda_j = \mu_j+1$,
$\lambda_k = \mu_k-1$,
$\lambda_i = \mu_i$ for all $i \not\in \{j, k\}$ and
one of the two cases holds:
\begin{enumerate}
\item $k = j+1$;
\item $\mu_j = \mu_k$ (or, equivalently, $\lambda_j = \lambda_k+2$).
\end{enumerate}
\end{lemma}

With the notation of this lemma, 
let $\lambda \succ_1 \mu$ 
(resp. $\lambda \succ_2 \mu$) denote the fact that
$\lambda \succ \mu$ with $k = j+1$ (resp. $\mu_j = \mu_k$).

The following result shows that the two rewriting rules on partitions of integers that occurs in Lemma~\ref{L_covers_Ln} can be translated to words. But this translation depends on the partition $\partition_p(w)$ or $\partition_s(w)$ that we choose. Actually the two translations are dual.

Let $\rightarrow_{a}$ and $\rightarrow_{b}$ be the two rewriting rules defined as follows. For $\alpha \in \{a, b\}$,
$u \rightarrow_\alpha v$ if there exist words $x$ and $y$, and
an integer $m$ such that $u = x ab\alpha^m ba y$
and $v = x ba\alpha^m ab y$.

\begin{lemma}
\label{L_cor_rewriting}
For any words $u$ and $v$:
\begin{enumerate}
\item $u \rightarrow_a v$ if and only if
$\partition_p(v) \succ_1 \partition_p(u)$
if and only if
$\partition_s(u) \succ_2 \partition_s(v)$.
\item $u \rightarrow_b v$ if and only if
$\partition_p(v) \succ_2 \partition_p(u)$
if and only if
$\partition_s(u) \succ_1 \partition_s(v)$.
\end{enumerate}
\end{lemma}

\begin{proof}
Let $u$ and $v$ be words.
For $1 \leq i \leq |u|_b$,
let $p_i$ be the word such that $p_ib = \pref_{i, b}(u)$.
For $1 \leq i \leq |v|_b$,
let $p_i'$ be the word such that $p_i'b =\pref_{i, b}(v)$.

Assume first that $u \rightarrow_a v$: 
$u = xaba^mbay$ and 
$v = xbaa^maby$
for some words $x$ and $y$ and some integer $m$.

Let $j = |xab|_b = |xb|_b$.
Observe that $p_jb=xab$, $p_j'b = xb$, $p_{j+1}b= xaba^mb$,
$p_{j+1}'b= xbaa^mab$.
So $|p_j'b|_a = |p_jb|_a-1$,
$|p_{j+1}'b|_a = |p_{j+1}'b|_a+1$
and, for $i \not\in \{j, j+1\}$,
$|p_ib|_a = |p_i'b|_a$.
By definition, 
$\partition_p(u) = (\ldots, |p_{j+1}b|_a, |p_jb|_a, \ldots)$
and 
$\partition_p(v) = (\ldots, |p_{j+1}b|_a+1, |p_j'b|_a-1, \ldots)$
(with the left and right ``\ldots" being equals).
Hence $\partition_p(v) \succ_1 \partition_p(u)$. 

The converse, that is the fact that 
$\partition_p(v) \succ_1 \partition_p(u)$ implies $u \rightarrow_b v$, is quite verbatim.
Hypothesis $\partition_p(u) \prec_1 \partition_p(v)$
implies that for some integer $j$, 
$\partition_p(u) = (\ldots, |p_{j+1}b|_a, |p_jb|_a, \ldots)$
and 
$\partition_p(v) = (\ldots, |p_{j+1}'b|_a, |p_j'b|_a, \ldots)$.
Hence 
$|p_{j+1}'b|_a = |p_{j+1}b|_a-1$,
 $|p_j'b|_a = |p_jb|_a+1$,
and $|p_ib|_a = |p_i'b|_a$ for $i \not\in \{j, j+1\}$.
Hence $u = xaba^mbay$ and  $v = xbaa^maby$
with $xab = p_jb$, $x_b = p_j'b$, $xaba^mb = p_{j+1}b$
and $xaba^mab =p_{j+1}'b$ 
for some words $x$ and $y$ and some integer $m$.

We still assume that $u \rightarrow_a v$, 
$u = xaba^mbay$ and 
$v = xbaa^maby$.
For $1 \leq i \leq |u|_a$,
let $s_i$ be the word such that $as_i = \suff_{i,a}(u)$.
For $1 \leq i \leq |v|_a$,
let $s_i'$ be the word such that $as_i' = \suff_{i,a}(v)$.
Let remember that
$\partition_s(u) = (|as_{|u|_a}|_b, \ldots, |as_1|_b)$ and
$\partition_s(v) = (|as_{|u|_a}'|_b, \ldots, |as_1'|_b$)
Let $j = |y|_a+1$ Observe that:
\begin{itemize}
\item $as_j = ay$, $as_j' = aby$
\item $as_{j+i} = a^ibay$, $as_{j+1}' = a^iaby$, for $1 \leq i \leq m$,
\item $as_{j+m+1} = aba^mbay$, $as_{j+m+1}' = aa^maby$.
\end{itemize}
Thus
$\partition_s(u) = ( \ldots, 2+|y|_b, |y|_b+1, \ldots, |y|_b+1, |y|_b, \ldots)$
and
$\partition_s(u) = ( \ldots, 1+|y|_b, |y|_b+1, \ldots, |y|_b+1, |y|_b +1, \ldots)$:
more precisely
$|as_{j+m+1}|_b = |y|_b+2$,
$|as_{j+i}|_b = |y|_b+1$ for $1 \leq i \leq m$,
$|as_j|_b = |y|_b$,
$|as_{j+m+1}'|_b = |y|_b+1$,
$|as_{j+i}|_b = |y|_b+1$ for $1 \leq i \leq m$,
$|as_j|_b = |y|_b+1$ (and $|as_k|_b = |as_k'|_b$ for $k \not\in [j, j+m+1]$).
Hence $\partition_s(u) \succ_2 \partition_s(v)$.

The converse follows in a similar way than with the hypothesis 
$\partition_p(v) \succ_1 \partition_p(u)$.

The proof of the second item can be proved similarly.
\end{proof}

\begin{remark}
In their study of partitions of integers in the context of 
Parikh word representable graphs, L.~Mathew et al. \cite{Mathew_Thomas_Bera_Subramanian2019AAM}
defined the dual of a word $w$ as the image of the reverse of $w$ by the anti-morphic involution that maps $a$ to $b$. This corresponds in the geometrical representation of words 
to consider the word obtained by a symmetry using the main diagonal of $\rect(w)$.
It may be observed that the two items of Lemma~\ref{L_cor_rewriting}
are the same up to this duality.
\end{remark}

Let us also remark that if $u \rightarrow_a v$ or
$u \rightarrow_b v$ then $u \rightarrow v$. 
Actually it is a simple observation that 
for a word $x$, neither $ab$ nor $ba$ is a factor of $x$ if and only if in $x \in a^*\cup b^*$.
Hence relations $\rightarrow_a$ and $\rightarrow_b$ 
is the restriction of the relation $\rightarrow$ used on 
factors $abxba$ with $x$ containing no occurrence of $ab$ nor $ba$.

Next lemma shows that this restriction does not modify the 
general rewriting. Let $\rightarrow_{spa}$ be the union of 
relations $\rightarrow_a$ and $\rightarrow_b$ 
($spa$ stands for special palindromes; Remark~\ref{R_spa} explains this notation).
The relation $\rightarrow_{spa}^*$ denotes the
transitive closure of relation  $\rightarrow_{spa}$.

\begin{lemma}
\label{L_equivalence_spa_rightarrow}
For any words $u$ and $v$,
$u \rightarrow_{spa}^* v$ if and only if $u \rightarrow^* v$.
\end{lemma}

\begin{proof}
We have already observed that $u \rightarrow_{spa} v$ 
implies $u \rightarrow v$. Hence $u \rightarrow_{spa}^* v$ 
implies $u \rightarrow^* v$.

To prove the converse 
we show that if $u \rightarrow v$ then
$u \rightarrow_{spa}^* v$: this implies that if $u \rightarrow^* v$ then
$u \rightarrow_{spa}^* v$. We act by induction on $|u| = |v|$. There is nothing to prove when $u = v$, especially when $u$ 
is the empty word.

So assume that $u = xabybaz$ and $v = xbayabz$.
If $x \neq \varepsilon$ or $z \neq \varepsilon$, by inductive hypothesis
$abyba \rightarrow_{spa}^* bayab$.
Thus
$xabybaz \rightarrow_{spa}^* xbayabz$
(actually $\rightarrow_{spa}^*$ is a congruence).
So assume $x = z = \varepsilon$.

If $y$ contains no occurrence of $ab$ nor of $ba$
then $u \rightarrow_{spa} v$. 
Otherwise $y = pabs$ or $y = pbas$ for some words $p$ and
$s$.
Assume first $y = pabs$.
Observe that $|absba|<|abyba|$ 
and $absba \rightarrow basab$. By inductive hypothesis
$absba \rightarrow_{spa}^* basab$.
Hence $abpabsba \rightarrow_{spa}^* abpbasab$.
Since $|abpba| < |abyba|$, we also have by induction
$abpba \rightarrow_{spa}^* bapab$
and so $abpbasab \rightarrow_{spa}^* bapabsab$.
Hence $u = abpabsba \rightarrow_{spa}^* bapabsab=bayab = v$.

The case $y = pbas$ can be treated similarly.
\end{proof}

Lemma~\ref{L_equivalence_spa_rightarrow} and Theorem~\ref{T_lattice} implies that the graph
$([w]_{\sim_2}, \rightarrow_{spa})$ has also a lattice structure.
Hence using Lemmas~\ref{L_bijection_partition}, \ref{L_bijection_partition_dual} and \ref{L_equivalence_spa_rightarrow}, we get next result
in which $L_{n,k,\ell}$ is the number of partitions of the integer $n$ in $k$ parts and the greatest part bounded by $\ell$.

\begin{theorem}
\label{T_lattice_isomorphism}
Let $w$ be a word over $\{a, b\}$.
Let $n = \binom{w}{ab}$. the graph
$([w]_{\sim_2}, \rightarrow_{spa})$ is a lattice 
isomorphic both to $L_{n, |w|_a, |w|_b}$ and $L_{n, |w|_b, |w|_a}$.
Its least element is $\init(w)$ and
its greatest element is $\final(w)$.
\end{theorem}

\begin{remark}
\label{R_spa}
The notation $\rightarrow_{spa}$ 
comes from the fact that words in $aba^*ba \cup abb^*ba$ are
special palindromes.
This relation has a connection
with palindromic amiability introduced by A.~Atanasiu et al. \cite{Atanasiu_Martin-Vide_Mateescu2002FI} (see also \cite{Fosse_Richomme2004IPL}).  
Set $u \equiv_{pal} v$ if $u = x \pi_1 y$ and $u = x \pi_2 y$
with $\pi_1$, $\pi_2$ palindromes and 
$(|\pi_1|_a, |\pi_1|_b) = (|\pi_2|_a, |\pi_2|_b)$.
Let $\equiv_{pal}^*$ the transitive closure of $\equiv_{pal}$.
Two words $u$ and $v$ are said palindromic amiable if $u \equiv_{pal}^* v$. In \cite{Atanasiu_Martin-Vide_Mateescu2002FI}, it is proved: 

\begin{theorem}\cite{Atanasiu_Martin-Vide_Mateescu2002FI}
\label{T_palindromic_amiability}
Two binary words have the same Parikh matrix (so are 2-binomially equivalent) if and only if they are  palindromic amiable.
\end{theorem}
\end{remark}

\begin{remark}
\label{rem_longest_paths}
As shown in the proof of Lemma~\ref{L_equivalence_spa_rightarrow},
$u \rightarrow_{spa} v$ implies $u \rightarrow v$
and $u \rightarrow v$ implies $u \rightarrow_{spa}^* v$.
So longest paths in the graph $([w]_{\sim_2}, \rightarrow_{spa})$
are longest paths in $G_2(w)$. See \cite{Greene_Kleitman1986EJC} for more information on these longest paths.
\end{remark}

The following result follows as a corollary of Theorems~\ref{T_equivalence_reecrite_sim2}
and Lemma~\ref{L_equivalence_spa_rightarrow}.

\begin{corollary}
\label{cor_charac_2bin_spa}
Two binary words $u$ and $v$ are 2-binomially equivalent
if and only if $u \equiv_{spa} v$.
\end{corollary}

\section{\label{sec_fair_words}Fair words}

\subsection{\label{subsec_fair_words_base}Definition, examples and basic properties}

Following A.~\cerny\ \cite{Cerny2009JALC}, 
a word $w$ over an arbitrary alphabet $A$
is \textit{fair} if $\binom{w}{ab} = \binom{w}{ba}$ for all pairs $(a, b)$ of letters. 
A.~\cerny\  justified his terminology as follows.
"\textit{Imagine members of $k \geq 1$ rivaling groups $a_1$, $a_2$, $a_k$ want to pass a narrow door. In which order they should pass? A solution to this problem is fair if, for any two distinct groups $a_i$, $a_j$, the number of member pairs, where a member of $a_i$ precedes a member of $a_j$, is the same as of those where the order is reversed. Any passing order can be denoted be a word on the alphabet $\Sigma = \{a_1, a_2, \ldots, a_k\}$, containing as many occurrences of $a_i$ as there are members if the group $a_i$. A word describing a fair solution will be called fair"}.
As mentioned in \cite{Cerny2007DAM} (a paper written after \cite{Cerny2009JALC}), fair words were already studied in 1979 in the binary case by H.~Prodinger \cite{Prodinger1979DM}. 
Prodinger's approach was different as \cerny's approach
since H.~Prodinger
considered fair words as a particular generalization
of unrestricted Dyck words.
Such an unrestricted Dyck word over $\{a, b\}$ 
is a word having the same number of occurrences of each letter, 
that is $|w|_a = |w|_b$ or $\binom{w}{a} = \binom{w}{b}$.
H.~Prodinger defined, for any words $x$ and $y$, the language
$\displaystyle D(x, y) = 
\left\{w \in \{a, b\}^* \mid \binom{w}{x} = \binom{w}{y}\right\}$
but studied only the language $D(ab, ba)$ which is the set of fair words over $\{a, b\}$.

As an immediate consequence of Formula~\eqref{eq_relation_base1}, we have

\begin{lemma}
\label{L_carac_fair_with_value_binom}
A word $w$ over an alphabet $A$ is fair if and only if, for all different letters $a$ and $b$, $\binom{w}{ab} = \frac{|w|_a|w|_b}{2}$.
\end{lemma}

\input{grille_6x7_exemple_illustration_fair_article.tex}

Figure~\ref{fig_fair_word_example} shows an example of binary fair word
in our context of geometrical representation. 
It is useful observing that, as a reformulation of the definition using  Lemmas~\ref{L_interpretation_binom_ab} and \ref{L_interpretation_binom_ba}), a binary fair word $w$ is a word whose line representation cuts into two parts of same area the rectangle of height $|w|_a$ and width $|w|_b$: we find again the fact that $\binom{w}{ab} = \frac{|w|_a|w|_b}{2}$.

Since $\binom{w}{ab} = \frac{|w|_a|w|_b}{2}$ and $\binom{w}{ab}$ is an integer,
it follows that no fair word $w$ exists with $|w|_a$ and $|w|_b$ both odd. Words $a^kb^\ell a^k$ and $b^ka^\ell b^k$
show that given an even number of occurrences of $a$ or of $b$ one can construct a fair word. Actually as mentioned by A.~\cerny\ \cite{Cerny2009JALC} and by H.~Prodinger (see last line of Page 270 in \cite{Prodinger1979DM}), all palindromes are fair. Figure~\ref{fig_fair_word_example} shows an example of non palindromic fair words. 
A.~\cerny\ mentioned that smallest non-palindromic fair words are of length $7$ ($abbbaab$, $baabbba$,
$baaabba$ and $abbaaab$).


Generalizing previous examples of fair words, one can observe that:
\begin{itemize}
\item A word $w \in b^*a^*b^*$ is fair if and only if $w = b^na^mb^n$ for some $n, m$;
\item A word $w \in b^*a^*ba^*b^*$ is fair if and only if $w = b^na^mba^mb^n$ for some $n, m$.
\end{itemize}
It is also immediate from the definition of fair words that 
for two binary words $u$ and $v$ with $u \sim_2 v$: $u$ is fair
if and only if $v$ is fair.
Consequently, as a consequence of
Theorem~\ref{T_lattice} and Lemma~\ref{L_lang_sans_ab...ba} (and the definition of $\init(w)$ and 
$\final(w))$, we have
\begin{lemma}
Are equivalent for a binary word $w$:
\begin{itemize}
\item $w$ is fair
\item $\init(w)$ is a palindrome
\item $\final(w)$ is a palindrome
\end{itemize}
\end{lemma}

\begin{remark}
For a word $w$, the fact that $\init(w) \rightarrow^* w$
can be interpreted using the previous lemma:
every fair word can be constructed from a palindrome by finitely many substitutions of some proper factor.
This answers Problem 23 in \cite{Cerny2009JALC}.
\end{remark} 

\begin{remark}
\label{rem_equiv2_fair_words}
Since 
$\binom{w}{ab} = \frac{|w|_a|w|_b}{2}$ for any fair word $w$ and any different letters $a$ and $b$, 
we have for any fair words $u$ and $v$

\centerline{$u \sim_2 v$ if and only if $\Psi(u) = \Psi(v)$.}
\end{remark}

To end this part let us show a new characterization of the $2$-binomial equivalence $\sim_2$ for binary words using a rewriting system allowing to replace any fair word by another fair word.
Let define the relation $\equiv_{fair}$ on words over $\{a, b\}$ as follows:
$u \equiv_{fair} v$ 
if there exist words $x$, $y$ and fair words $\pi_1$ and $\pi_2$
such that $u = x \pi_1 y$, $v = x \pi_2 y$ and $\Psi(\pi_1) = \Psi(\pi_2)$. 
Let $\equiv_{fair}^*$ be the transitive closure of $\equiv_{fair}$.

\begin{lemma}
\label{L_equivalence_sim2_et_equivFair}
For any words $u$ and $v$ over $\{a, b\}$, 
$u \sim_2 v$ if and only if
$u \equiv_{fair}^* v$.
\end{lemma}

\begin{proof}
If $u \sim_2 v$ then, 
by Theorem~\ref{T_palindromic_amiability},
$u \equiv_{pal}^* v$. Since any palindrome is fair,
it follows the definition that $x \equiv_{pal} y$ implies 
$x \equiv_{fair} y$ and that $x \equiv_{pal}^* y$ implies 
$x \equiv_{fair}^* y$. Hence $u \equiv_{fair}^* v$.

Conversely, assume that $u \equiv_{fair}^* v$.
There exist words $x$, $y$ and fair words, $\pi_1$ and $\pi_2$
such that $u = x \pi_1 y$, $u = x \pi_2 y$ and $\Psi(\pi_1) = \Psi(\pi_2)$.
Observe $\Psi(x\pi_1y) = \Psi(x\pi_2y)$.
Since $\pi_1$ and $\pi_2$ are {fair} words with $\Psi(\pi_1) = \Psi(\pi_2)$, $\binom{\pi_1}{ab} = \frac{|\pi_1|_a|\pi_1|_b}{2}
= \binom{\pi_2}{ab}$.
Consequently Formula~\eqref{eq_binom_3_words} implies
that $\binom{x\pi_1y}{ab} = \binom{x\pi_2y}{ab}$.
Hence $u \sim_2 v$.
\end{proof}

Observe that Lemma~\ref{L_equivalence_sim2_et_equivFair}
cannot be extended to arbitrary alphabets (as Theorem~\ref{T_palindromic_amiability}). 
For words $u = abcacab$ and $v = caabbac$,
we have $u \sim_2 v$ but $u \not\equiv_{fair} v$.
Indeed the only fair factors of $u$ are palindromes that are of length at most $3$.
For such a palindrome $\pi$, $[\pi]_{\sim_2}$ is a singleton
and so there is no other fair words $\pi'$ with $\Psi(\pi') = \Psi(\pi)$ 
by Remark~\ref{rem_equiv2_fair_words}.

\subsection{\label{subsec_language_properties}Language properties}

Observe as direct consequences of the definitions
that the set of fair words is closed by permutation of letters (exchange in the binary case) and by mirror image. In 1979, H.~Prodinger \cite{Prodinger1979DM} proved that the language of binary fair words is not context-free. 
Independently A.~Salomaa \cite{Salomaa2007FI} proved in 2007 this 
result for the language  of fair words over arbitrary alphabets.
He also stated that the language of fair words is context-sensitive.

In \cite{Prodinger1979DM}, H.~Prodinger studies the syntactic congruence of the language of binary fair words and he proved 
that, in the context of binary words, the syntactic congruence of the language 
of fair words is the $2$-binomial equivalence (see his Theorem 2 and the remark that follows on page 271). 
Actually his proof works quite verbatim for arbitrary alphabets. 
We explain this below.

Given a language $L$ over an alphabet $A$, that is a set of words over $A$, the \textit{syntactic congruence} $\sim_L$ 
is defined by $x \sim_L y$ if and only if 
for all $u, v \in A^*$, $uxv \in L$ if and only if $xvy \in L$.
For letters $a$ and $b$, let 
$\displaystyle \Delta_{ab}(w) = \binom{w}{ab}-\binom{w}{ba}$.
Let ${\cal F}$ be the language of fair words over $\{a, b\}$: 
it is the set of word $w$ such that $\Delta_{ab}(w) = 0$
(This language is also denoted $D(ab,ba)$ in \cite{Prodinger1979DM} as mentioned in Section~\ref{subsec_fair_words_base}).

\begin{theorem}(\cite{Prodinger1979DM} for the binary case)
\label{T_Prodinger_syntactic_congruence_extended}
Let $A$ be an alphabet. Are equivalent for any words $u$ and $v$:
\begin{enumerate}
\item $u \sim_{\cal F} v$;
\item for all pairs of letters $(a, b)$, $|u|_a = |v|_a$ and
$|u|_b = |v|_b$ and $\Delta_{ab}(u) = \Delta_{ab}(v)$;
\item for all pairs of letters $(a, b)$, $\pi_{a,b}(u) \sim_2 \pi_{a,b}(v)$.
\end{enumerate}
\end{theorem}

\begin{proof}
Proof of $1 \Rightarrow 2$. We copy the proof of H.~Prodinger
just adapting it for arbitrary alphabets instead of binary alphabets.

Let $u$, $v$ and $x$ be words over $A$ such that
$u \sim_{\cal F} v$. 
Since $\sim_{\cal F}$ is a congruence,
$ux\widetilde{(ux)} \sim_{\cal F} vx\widetilde{(ux)}$
and 
$\widetilde{(ux)}ux \sim_{\cal F} \widetilde{(ux)}vx$.
Since $ux\widetilde{(ux)}$ is a palindrome, it is fair.
Hence $vx\widetilde{(ux)}$ is fair.
Similarly $\widetilde{(ux)}ux$ and  $\widetilde{(ux)}vx$ are fairs.
Thus using Formula~\eqref{eq_binom_2_words} and the fact that
$\Delta(\widetilde{w}) = -\Delta(w)$ for any word $w$, we observe that for any pair of letters $(a, b)$:
$$0 = \Delta_{ab}(vx\widetilde{(ux)}) = 
\Delta_{ab}(vx) -\Delta_{ab}(ux)
+|vx|_a |ux|_b -|vx|_b |ux|_a$$
and 
$$0 = \Delta_{ab}(\widetilde{(ux)}vx) = \Delta_{ab}(vx)-\Delta_{ab}(ux) 
+|ux|_a |vx|_b -|ux|_b |vx|_a$$
Thus adding and subtracting these equations, we obtain:
$\Delta_{ab}(ux) = \Delta_{ab}(vx)$ and
$|ux|_a |vx|_b = |ux|_b |vx|_a$.

With $x = \varepsilon$, we get $\Delta_{ab}(u) = \Delta_{ab}(v)$
and $|u|_a|v|_b = |u|_b|v|_a$.

With $x = a$, we get $(|u|_a+1)|v|_b = |u|_b(|v|_a+1)$ and so
$|u|_b = |v|_b$.

With $x = b$, we get $|u|_a(|v|_b+1) = (|u|_b+1)|v|_a$ and so
$|u|_a = |v|_a$.

\medskip

Proof of $2 \Rightarrow 3$.
By Formula~\eqref{eq_relation_base1}, 
$\binom{w}{ab} = \frac{1}{2}(\Delta_{ab}(w) + |w|_a|w|_b)$. 
Thus 
$\Delta_{ab}(u) = \Delta_{ab}(v)$
implies that $\binom{u}{ab} = \binom{v}{ab}$, and so
$2 \Rightarrow 3$ follows directly the definition of $\sim_2$.
Indeed $\pi_{ab}(u) \sim_2 \pi_{ab}(v)$ just means
$|u|_a = |v|_a$, $|u|_b = |v|_b$ and $\binom{u}{ab} = \binom{v}{ab}$.

\medskip

Proof of $3 \Rightarrow 1$.
If Item 3 holds, we have 
$|u|_a = |v|_a$, $|u|_b = |v|_b$ and $\binom{u}{ab} = \binom{v}{ab}$ for all pairs of letters $(a, b)$. 
Hence,
for any words $x$ and $y$ and any pair of letters $(a, b)$,
$|xuy|_a = |xvy|_a$, $|xuy|_b = |xvy|_b$ and, by Formula~\eqref{eq_binom_3_words}, $\binom{xuy}{ab} = \binom{xvy}{ab}$. So $xuy$ is fair if and only if $xvy$ is fair: $xuy \sim_{\cal F} xvy$.
\end{proof}

\subsection{\label{subsecSturm}Fair balanced words}

As already mentioned,
a fair word $w$ is a word whose line representation cuts into two parts of same area the rectangle of height $|w|_a$ and width $|w|_b$.
In usual geometry, the diagonal segment is one of the more
natural way to cut a rectangle into two parts of equal areas. 
It is well-known that infinite straight lines are represented by 
Sturmian sequences and that segments are represented by balanced words, that is, factors of Sturmian sequences (see \cite{Lothaire2002book} for instance). 
So a natural question is which are the balanced fair words.
We answer this question proving:

\begin{theorem}
\label{T_balanced}
A balanced word over $\{a, b\}$ is fair 
if and only if
it is a palindrome.
\end{theorem}

Let recall that a word $w$ over $\{a, b\}$ is \textit{balanced}
if for any factors $u$ and $v$ of $w$ of same length,
$||u|_a-|v|_a| \leq 1$. 
Any factor of a balanced word is also balanced. In the next proofs, $\first(w)$ and
$\last(w)$ denote respectively the first and the last letter of the word $w$. We need the following intermediary result.

\begin{lemma}
\label{L_fair_reduction}
Given a word $w$ over an alphabet $A$ and $\alpha \in A$, $w$ is fair if and only if $\alpha w\alpha$ is fair
\end{lemma}

\begin{proof}
Let $a$ and $b$ be two different letters in $A$.
If $\alpha \not\in w$, $\binom{w}{ab} = \binom{\alpha w\alpha}{ab}$
and $\binom{w}{ba} = \binom{\alpha w\alpha}{ba}$ since any occurrence of $ab$ or $ba$ in $\alpha w\alpha$ necessarily occurs in $w$.
Assume $\alpha = a$. 
Observe that 
$\binom{awa}{ab} = \binom{w}{ab} + |w|_b$
and $\binom{awa}{ba} = \binom{w}{ba} + |w|_b$.
Hence $\binom{w}{ab} = \binom{w}{ba}$
if and only if $\binom{a wa}{ab} = \binom{a wa}{ba}$.
Similarly, we can prove the same result when $\alpha = b$. The lemma follows the definition of fairness.
\end{proof}

\begin{proof}[Proof of Theorem~\ref{T_balanced}.]
Since any palindrome is a fair word, we only have to prove
that any fair balanced word is a palindrome. 
 Assume by contradiction that
 there exists a fair balanced word $w$ which is not a palindrome.
 Choose $w$ of minimal length.
 
 If $w = aw'a$ or $w = bw'b$, then $w'$ is balanced 
 since any factor of a balanced word is balanced.
 By Lemma~\ref{L_fair_reduction}, 
 the word $w'$ is also fair. 
 Observe finally also that $w'$ is not a palindrome since $w$ is not a palindrome.
 This contradicts the choice made for $w$.
 Thus $w = aw'b$ or $w = bw'a$.
 
 Since we will consider simultaneously $w$ and $\widetilde{w}$,
 without loss of generality,
 we can assume that $w = aw'b$ and $\widetilde{w} = b\widetilde{w'}a$.
Let $((u_i, v_i))_{1 \leq i \leq k}$ denote
 the maximal $\Psi$-decomposition of $(w, w')$.
 Observe that $\first(u_1) = a$ and $\first(v_1)=b$.

 Assume that for all $i$, $1 \leq i \leq k$, $u_i = v_i$
 or both $\first(u_i) = a$ and $\first(v_i) = b$.
 By Lemma~\ref{L_Psi_difference},
 $\binom{w}{ab}-\binom{\widetilde{w}}{ab}= \sum_{i = 1}^k \binom{u_i}{ab}-\binom{v_i}{ab}$.
 Since each $(u_i, v_i)$ is $\Psi$-undecomposable by 
 Lemma~\ref{L_Psi_value} and by hypothesis above,
 $\binom{u_1}{ab}-\binom{v_1}{ab} > 0$,
 and, for all $i$, $2 \leq i \leq k$,
 $\binom{u_i}{ab}-\binom{v_i}{ab} \geq 0$.
 Hence $\binom{w}{ab}-\binom{\widetilde{w}}{ab} > 0$.
 But $\binom{\widetilde{w}}{ab}=\binom{w}{ba}$ and, since $w$ is fair,
 $\binom{w}{ba}=\binom{w}{ab}$. We have a contradiction.
 
 So there exists an integer $\ell$ such that
 $2 \leq \ell \leq k$, $\first(u_\ell) = b$ 
 and $\first(v_\ell) = a$.
 Let $x = u_2 \cdots u_{\ell-1}$ and $y = v_2 \cdots v_{\ell-1}$.
 The words $u_1xb$ and $v_1ya$ are respectively factors of $w$ and $\widetilde{w}$.
 
 Now consider the $\Psi$-undecomposable pair $(u_1, v_1)$.
 The last letters of $u_1$ and $v_1$ 
 are different because of undecomposability.
If $u_1$ ends with $a$
and $v_1$ ends with $b$,
then letting $i = |u_1|_b = |v_1|_b$,
$|\pref_{i,b}(u_1)|_a < |u_1|_a = |v_1|_a = |\pref_{i,b}(v_1)|_a$.
This contradicts Lemma~\ref{L_Psi_value}(1).
So $u_1$ ends with $b$ and $v_1$ ends with $a$.
 Thus $bxb$ and $aya$ are respectively factors
 of $w$ and $\widetilde{w}$. So $bxb$ and $a\widetilde{y}a$ are factors of $w$.
 Since 
 $\Psi(x)= \sum_{i=2}^\ell \Psi(u_i)$,
 $\Psi(y)= \sum_{i=2}^\ell \Psi(v_i)$,
 and, for $2 \leq i \leq \ell-1$, $\Psi(u_i) = \Psi(v_i)$ (by definition of the $\Psi$-decomposition),
 we have $\Psi(x) = \Psi(y)$.
 This implies that $|x|=|\widetilde{y}|$ and $|bxb|=|a\widetilde{y}a|$: this contradicts the 
 balancedness of $w$. This ends the proof since this last contradiction raises from the hypothesis
 $w$ is fair and balanced.  
 \end{proof}

From Remark~\ref{rem_equiv2_fair_words},
we know that any equivalence class of a fair word $w$ over $\{a,b\}$ by $\sim_2$
is determined by the numbers of occurrences of $a$ and $b$.
We know also that at least one of these numbers must be even.
Next proposition shows that 
any equivalence class of a fair word over $\{a,b\}$
contains a palindromic balanced word.
For $\alpha \in \{a, b\}$,
let $L_\alpha$ be the free monoid morphism defined on $\{a, b\}^*$
by $L_\alpha(\alpha) = \alpha$ and $L_\alpha(\beta) = \alpha\beta$
for $\beta \neq \alpha$.

\begin{proposition}
\label{P_existence_palindrome_fair}
Given two integers $k \geq 0$ and $\ell \geq 0$
with at least one even, there exists a palindromic balanced fair $w$ with $|w|_a = k$
and $|w|_b = \ell$.
\end{proposition}

\begin{proof}
The proof acts by induction on $\max(k, \ell)$

First if $k = \ell$, then the two integers are even.
There exists an integer $m$ such that $k = 2m$.
The word $(ab)^m(ba)^m$ is a palindromic balanced words with
$k$ occurrences of $a$
and $\ell$ occurrences of $b$.
Second if $k = 0$ (resp. $\ell = 0$), the word $b^\ell$ (resp. $a^k$)
answers the proposition.

Assume now that $1 \leq \ell < k$.
Since at least one of the integers $k$ and $\ell$ is even,
also at least one of the two integers $k-\ell-1$ and $\ell$
is even. By inductive hypothesis,
there exists a palindromic balanced word $\pi$
such that $|\pi|_a = k-\ell-1$ and $|\pi|_b = \ell$.

By \cite{DroubayJustinPirillo2001TCS}, it is known that $L_a(\pi)a$ is a palindrome (this can also be checked easily).
Also $L_a(\pi)a$ is balanced. 
Indeed if it is not balanced, it contains two factors $axa$ and $bxb$. 
One can then verify that $x = L_a(y)a$ for some word $y$, and further, both words $aya$ and $byb$ are factors of $\pi$ contradicting its balancedness. 
Finally observe
that $|L_a(\pi)a|_a = |\pi|_a+|\pi|_b+1 = k$ and $|L_a(\pi)a|_b = \ell$.
The word $L_a(w)a$ answers the proposition.

The case $1 \leq \ell < k$ can be treated similarly using $L_b$ instead of $L_a$.
\end{proof}

\subsection{\label{subsec_number_fair_words}On the number of fair words}

One aim of A.~\cerny\ in \cite{Cerny2009JALC}
is the study of the number of fair words.
In \cite{Cerny2008IJFCS}, he mentioned 
Prodinger's paper \cite{Prodinger1979DM} 
in which the following results were already stated.
Let $\sigma$ be defined by $\sigma(a) = +1$, $\sigma(b) = -1$.

\begin{lemma}[\cite{Prodinger1979DM}]
For each word $w = a_1\cdots a_n$ ($a_i \in \{a, b\}$),

$$\binom{w}{ab} - \binom{w}{ba} = \frac{1}{2} \sum_{k = 1}^n \sigma(a_k) (n+1-2k)$$
\end{lemma}

\begin{corollary}[\cite{Prodinger1979DM}]
The number of fair words of length $n$ is\\
the number of solutions $(\epsilon_1, \ldots, \epsilon_n)$ with $\epsilon_i \in \{-1, +1\}$ of
$$\sum_{k = 1}^n \epsilon_k (n+1-2k) = 0$$
\end{corollary}

\begin{theorem}[\cite{Prodinger1979DM}]
The number $f(n)$ of fair words of length $n$ 
verifies
$$f(n) \sim 2^{2\lfloor(n-1)/2\rfloor+1}\left(\frac{3}{\pi}\right)^{1/2}\left\lfloor\frac{n}{2}\right\rfloor^{-3/2}$$
\end{theorem}

In \cite{Cerny2009JALC}, the numbers
of fair words of length $n$ over a $k$-ary alphabet
were provided when $k = 2$ and $n \leq 20$, $k = 3$ and $n \leq 16$, $k = 4$ and $n \leq 12$, $k = 5$ and $n \leq 10$.
When $k = 2$, theses values are: 1, 2, 2, 4, 4, 8, 8, 20, 18, 52,48, 152, 138, 472, 428, 1520, 1392, 5044, 4652, 17112, 15884.
Removing the first value 1, we get 
sequence A222955 in the On-line Encyclopedia of integer sequences \cite{OEISA222955}.
Here follows the description of this sequence:
``\textit{Number of nX1 0..1 arrays with every row and column least squares fitting to a zero slope straight line, with a single point array taken as having zero slop}".
On this page, when visited, no link was made with fair words
and the paper by H.~Prodinger was not cited.
We prove this link.

Here we consider the alphabet $\{0, 1\}$ instead of the alphabet
$\{a, b\}$. 
As mentioned, for instance, in \cite{wiki_least_squares},
the method of least squares is an old (more than two centuries) and well-known ``\textit{mathematical optimization technique
that aims to determine the best fit function by minimizing the sum of the squares of the differences between the observed values and the predicted values of the model}".
In our context, the model is a straight line with slope $\beta$ and intercept $\alpha$.
Let $S = \{ (x_i, y_i) \mid 1 \leq i \leq n\}$
be a set of points and let $f_S(\alpha, \beta) = 
\sum_{i = 1}^n (\alpha+x_i\beta-y_i)^2$.
The least square method (for this model) aims to
know the values $\alpha$ and $\beta$ for which $f_S(\alpha,\beta)$
is minimal. It is known that $f_S$ is minimal in an unique pair
$(\alpha,\beta)$ and this pair may be determined by the two local
conditions: $\partial_\alpha f_S(\alpha,\beta) = 0$ and $\partial_\beta f_S(\alpha,\beta) =0$.
In our problem, the set $S$ is $S(w) = \{(i, w_i) \mid 1 \leq i \leq n \}$:
$$f_{S(w)}(\alpha,\beta) = \sum_{i = 1}^n (\alpha+i\beta-w_i)^2$$.

\begin{theorem}
\label{T_least_squares}
Let $w$ be a word over $\{0,1\}$.
The minimization of $f_{S(w)}(\alpha, \beta)$ 
is obtained when $\beta = 0$
if and only if $w$ is fair. 
In this minimal case, $\alpha = \frac{|w|_1}{|w|}$.
\end{theorem}

\begin{proof}
Let $n$ denote the length of $w$: $w = w_1 \cdots w_n$ with $w_i \in \{0, 1\}$ for each $i$, $1 \leq i \leq n$.
We search the minimum of $f_{S(w)}$.
Let us make three preliminary observations.
\begin{itemize}
\item Observe that $\sum_{i=1}^n w_i = |w|_1$, and so,
$\partial_\alpha f_{S(w)}(\alpha,\beta) = 
\sum_{i=1}^n 2(\alpha+i\beta-w_i) = 
2n\alpha +2\beta\sum_{i=1}^ni -2|w|_1$.
\item Following Section~\ref{subsec_sum_positions},
let $S_1(w) = \sum_{i=1}^n iw_i$ 
be the sum of positions of $1$s in $w$.
We have
$\partial_\beta f_{S(w)}(\alpha, \beta) = 
\sum_{i = 1}^n 2i (\alpha+i\beta-w_i) =
\alpha n(n+1) + 2\beta \sum_{i = 1}^n i^2 - 2 S_1(w)$.
\item By Lemma~\ref{L_lien_Sb_et_binom_ab},
$S_1(w) = \frac{|w|_1(|w|_1+1)}{2} + \binom{w}{01}$.
Thus the word $w$ is fair
if and only if
 $\binom{w}{01} = \frac{|w|_0|w|_1}{2}$
if and only if 
$S_1(w) = \frac{|w|_1(|w|+1)}{2} = \frac{|w|_1(n+1)}{2}$.
\end{itemize}

Assume now that $\beta = 0$.
From $\partial_\alpha f_{S(w)}(\alpha, 0) = 0$, 
we get $\alpha = \frac{|w|_1}{n}$.
From $\partial_\beta f_{S(w)}(\alpha, 0) = 0$, 
we get $S_1(w) = \alpha \frac{n(n+1)}{2} = 
\frac{|w|_1(n+1)}{2}$. The word $w$ is fair.

Assume now that the word $w$ is fair.
From $\partial_\alpha f_{S(w)}(\alpha, 0) = 0$, 
we get $\beta n(n+1)= 2(|w|_1-n\alpha)$.
From $\partial_\beta f_{S(w)}(\alpha, 0) = 0$, 
we get $\frac{\beta}{3}n(n+1)(2n+1)=2S_1(w)-\alpha n(n+1)
=(n+1)(|w|_1-n\alpha)$.
Thus $|w|_1-n\alpha = \frac{\beta n(n+1)}{2}= \frac{\beta}{3}n(2n+1)$.
Since $n \geq 2$, $\beta=0$ and $\alpha = \frac{|w|_1}{n}$.
\end{proof}

The web page~\cite{OEISA222955} also contains the following conjecture:
``\textit{A binary word is counted iff it has the same sum of positions of 1's as its reverse, or, equivalently, the same sum of partial sums as its reverse. - Gus Wiseman, Jan 07 2023}".
This can be proved as a direct consequence of results of Section~\ref{subsec_sum_positions}.

\begin{theorem}
A word over $\{a,b\}$ is fair if and only if
it has the same sum of positions of $b$'s than its mirror image.
\end{theorem}

\begin{proof}
Since $\binom{\widetilde{w}}{ab} = \binom{w}{ba}$,
by Lemma~\ref{L_lien_Sb_et_binom_ab}, we have
$$S_b(w) = \frac{|w|_b(|w|_b+1)}{2} + \binom{w}{ab}$$
$$S_b(\widetilde{w}) = \frac{|w|_b(|w|_b+1)}{2} + \binom{w}{ba}$$
The theorem follows immediately.
\end{proof}

\section{\label{sec_conclusion}Conclusion and questions}

\subsection{About characterizations of binomial equivalence}

As indicated by the title,
the main purposes  of this paper
are: the presentation and study of the geometrical 
interpretation of the numbers of 
occurrences of subwords $ab$ and $ba$ in a word;
the study of the structure of a $2$-binomial equivalence class of binary words 
and the link with the lattice structure of the set
of partitions of the value $\binom{w}{ab}$; 
a study of the family of binary fair words.
In addition, 
we have listed many characterizations 
of $2$-binomially equivalent words.
Some has been previously stated 
in the context of the study of Parikh or precedence matrices\footnote{Let remember that, in Section~\ref{subsec_storing}, 
we pointed out 
that a Parikh matrix (or a precedence matrix) stores 
all the information to know the number
of occurrences in a binary word of all subwords of length $1$ and $2$}.
This is the case of the characterizations
stated in 
Theorems~\ref{T_equiv_precedence_matrix}, \ref{T_equiv_parikh_matrix_and_sum_positions},
\ref{T_equivalence_reecrite_sim2}
 and~\ref{T_palindromic_amiability}.
Characterization in the binary case of
Theorem~\ref{T_Prodinger_syntactic_congruence_extended}
was stated before this context of Parikh matrices.
Finally the characterizations provided
by Corollary~\ref{cor_charac_2bin_spa} and Lemma~\ref{L_equivalence_sim2_et_equivFair} are new.
Theorem~\ref{T_equiv_precedence_matrix} and \ref{T_Prodinger_syntactic_congruence_extended} are valid
over arbitrary alphabets.
In the context of Parikh matrices, 
many studies have considered generalizations of these characterizations, notably characterizations implying rewriting rules, but they have been made 
essentially on binary and ternary alphabets
and only partials result have been obtained. 
See for instance the article \cite{Mercas_Teh2025TCS}
and its references.  
On a ternary or larger alphabet, having the same Parikh matrix for two words is not equivalent to be $k$-binomially equivalent (whatever is $k$).
Hence it is natural to search for some generalizations of
previous characterizations for words
that are $k$-binomially equivalent ($k \geq 2$) over
arbitrary alphabets (instead of for words that are Parikh equivalent.
This seems to be a new question
but it is probably as difficult as in the context of Parikh matrices.
Another natural question is the existence of a geometrical interpretation for others binomial coefficient of words.

\subsection{About the maximal number of fair factors of a word}

Since any palindrome is fair,
fair words appear as a generalization of palindromes.
This simple remark may open many questions about
generalizing notions studied around palindromes. In this section
and the next one, we 
provide some examples of such questions. 

The word $abbbaab$ was 
mentioned as been one of the smallest non-palindromic fair words.
It contains $9$ fair words:
$\varepsilon$, $a$, $aa$, $b$, $bb$, $bbb$, $baab$, $abbba$ and $abbbaab$ itself. 
This shows that the following result \cite{DroubayJustinPirillo2001TCS} by X.~Droubay, J.~Justin and G.~Pirillo is no longer true while replacing ``palindrome factors" by ``fair factors":
\textit{A word $w$ has at most $|w|+1$ different palindrome factors}. This raise the questions:
What is the maximal numbers of fair words in a word $w$? Is the maximal number of fair words is obtained for rich words, that is, words that have a maximal number of palindromes? And only for rich words? (See, for instance, \cite{Glen_Justin_Widmer_Zamboni2009EJC} for more information on palindromic richness).

Let us quote that  the difference between the maximal numbers of fair and palindromic words is not bounded. Let us give an example. Using for instance the next Lemma~\ref{L_concat_2_fair_words}, it may be observed by induction that $w^k$ is fair for any fair word $w$ and any integer $k \geq 1$. Thus
$(abbbaab)^k$ is fair for any integer $k \geq 1$. It contains $k$ non-palindromic factors that are fair: the words $(abbbaab)^k$, $1 \leq i \leq k$ (it may be also observe that $(abbbaab)^k$ contains only finitely many palindromes since it contains no palindrome of length 6 or more).

Next example shows that the maximal number of fair
words of length $n$ is in $\Theta(n^2)$. Since the number of
palindromes in a word of length $n$ is in $O(n)$,
the maximal difference between the number of fair
words and palindromes is also in $\Theta(n^2)$. 
Let us first observe that:

\begin{lemma}
\label{L_concat_2_fair_words}
If $u$ and $v$ are fairs words over $\{a, b\}$, then 
the word $uv$ is fair
if and only if
$|u|_a|v|_b = |u|_b|v|_a$.
\end{lemma}

\begin{proof}
This follows from the definition of a fair word, Relation~\eqref{eq_binom_2_words} ($\binom{uv}{ab} = \binom{u}{ab}+\binom{v}{ab}+|u|_a|v|_b$) and the relation obtained exchanging the roles of the letters $a$ and $b$ ($\binom{uv}{ba} = \binom{u}{ba}+\binom{v}{ba}+|u|_b|v|_a$).
\end{proof}

Let recall that the Thue-Morse word is the fixed point of the morphism $\mu$ defined by $\mu(a) = ab$ and
$\mu(b) = ba$. The words $\mu^2(a) = abba$ and
$\mu^2(b)  = baab$ are palindromes and so they are fairs.
An easy induction on the length of $u$ over $\{a, b\}$ using Lemma~\ref{L_concat_2_fair_words} shows that $\mu^2(u)$ is a fair word for any word $u$. So the number of fair words in $\mu^2(w)$
is in $\Omega(\#\text{fact}(w))$ 
where $\text{fact}(w)$ is the set of factors of $w$.

Is is known \cite{Shallit1993GC}, 
that the maximal number of distinct factors 
in a word of length $n$ is $2^{k+1}-1+\left(\frac{n-k+1}{2}\right)$
where $2^k+k-1 \leq n < 2^{k+1}+k$.
This number is in $\Theta(n^2)$.
J.~Shallit  provides examples of
words for which this upper bound is attained \cite{Shallit1993GC}.
From what precedes, for such a word $w$, 
the maximal number of distinct factors 
in $\mu^2(w)$ is also in $\Theta(|\mu^2(w)|)$.
Let remember that the number
of palindromes in $\mu^2(w)$ is bounded by $|\mu^2(w)|+1$.
Hence the maximal difference between 
the numbers of fair factors and of palindromes in a word $\mu^2(u)$ is in $\Theta(|\mu(u)|^2)$.

\subsection{Fair length}

The \textit{palindromic length} of a finite word $w$ was defined \cite{Frid_Puzynina_Zamboni2013AAM} as the minimal number of palindromes occurring in any decomposition of 
$w$ over palindromes. For instance, the palindromic length of $abbbaab$ is 3 since $abbbaab = abbba.a.b = a.bb.baab$ and, moreover, $abbbaab$ is neither a palindrome, neither the concatenation of two palindromes. Similarly we can define the \textit{fair length} of a nonempty word $w$ as the least integer $k$ such that $w = u_1\cdots u_k$ with $u_1$, \ldots, $u_k$ nonempty fair words. For instance, 
the fair length of $abbbaab$ is $1$, and,
the fair length of $abbbaababbabaab$ is $2$ since $abbabaab$ is fair (the palindromic length of this word is $3$). As a consequence of Theorem~\ref{T_balanced}, the fair length of any balanced word is its palindromic length. In addition to the algorithmic question of computing efficiently the fair length of a word, a natural question is to compute for each integer $n$, the maximal fair length of words of length $n$ as was done by O.~Ravsky for palindromic lengths in \cite{Ravsky2003JALC}. In \cite{Frid_Puzynina_Zamboni2013AAM}, 
A.~Frid, S.~Puzynina and L.Q.~Zamboni opened a conjecture concerning 
the palindromic lengths:
\textit{In every infinite word which is not ultimately periodic, the palindromic length of factors (version: of prefixes) is unbounded}. J.~Rukavicka has recently showed that this conjecture is true \cite{Rukavicka2026EJC}.
When replacing ``palindromic length" by ``fair length",
the result is no longer true as we show it now.

\begin{proposition}
\label{P_Thue_Morse_fair_lengths}
The fair lengths of factors of the Thue-Morse word are bounded by $4$.
\end{proposition}

\begin{proof}[Proof of Proposition~\ref{P_Thue_Morse_fair_lengths}]

Any factor of the Thue-Morse word is in the form $s\mu^2(u)p$ with $u$ a factor of the Thue-Morse word,
$s$ a proper suffix of $\mu^2(a)$ or $\mu^2(b)$ 
($s \in \{\varepsilon, a, ba, bba, b, ab, aab\}$) and 
$p$ a proper prefix of $\mu^2(a)$ or $\mu^2(b)$ 
($p \in \{\varepsilon, a, ab, abb, b, ba, baa\}$).

As already mentioned,
an easy induction on the length of $u$ over $\{a, b\}$ using Lemma~\ref{L_concat_2_fair_words} shows that $\mu^2(u)$ is a fair word. Since $a$, $b$, $bb$, $aa$ are fair words, the fair length of any factor
in $\{\varepsilon, a, ba, bba, b, ab, aab\}\mu^2(u)
\{\varepsilon, a, b\}$ 
or in $\{\varepsilon, a, b\}\mu^2(u)
\{\varepsilon, a, ab, abb, b, ba, baa\}$
is bounded by $4$.

We have to consider the fair length of words
in $\{ba, bba, ab, aab\}\mu^2(u)\{ab, abb, ba, baa\}$.
By Lemma~\ref{fig_2_words}, $a\mu^2(u)a$ and $b\mu^2(u)b$ are fair words, so the fair length of words
in $\{ba, bba\}\mu^2(u)\{ab, abb\}$
and in $\{ab, aab\}\mu^2(u)\{ba, baa\}$ is bounded by $3$.

It remains to consider the fair length of words in
$\{ba, bba\}\mu^2(u)\{ba, baa\}$ and of words in 
$\{ab, aab\}\mu^2(u)\{ab, abb\}$.
Up to an exchange of letters, we only have to consider the
fair length of words in
$\{ba, bba\}\mu^2(u)\{ba, baa\}$.

Consider the word $ba\mu^2(u)ba$. If $u \in \{\varepsilon, b, bb\}$, the fair length is at most $3$. Otherwise
$u = vax$ with $x \in \{\varepsilon, b, bb\}$.
Hence $ba\mu^2(u)ba = ba\mu^2(v)ab.ba\mu^2(x)ba$ and has fair length at most $4$ since by Lemma~\ref{fig_2_words},
$ba\mu^2(v)ab$ is fair and since the fair length of 
$ba\mu^2(x)ba$ is at most $3$ as we have just seen.

Similarly we can check that the fair length of
words $ba\mu^2(u)baa$, $bba\mu^2(u)ba$, $bba\mu^2(u)baa$
is bounded by $4$.
\end{proof}

Bound $4$ in Proposition~\ref{P_Thue_Morse_fair_lengths} is optimal: the factor $aababbaa$ of the Thue-Morse word has fair length 4 (it is the smallest one and all fair words occurring in any of its decompositions are palindromes since they are of length at most $3$).


\subsection{Miscellaneous}

In the context of infinite words,
X.~Droubay and J.~Pirillo introduced the notion of \textit{palindrome complexity} as the function that, given an infinite word $w$ counts, for each integer $n$, the number of distinct palindromes of length $n$ occurring in $w$ \cite{DroubayPirillo1999TCS} (See also \cite{Allouche_Baake_Cassaigne_Damanik2003TCS}).
They have characterized the palindrome complexity 
of Sturmian words, that is, aperiodic infinite binary words whose factors are balanced: they are words having exactly one palindrome of each even length and two of each odd length. It is natural to introduce the notion of \textit{fair complexity} as the function that, given an infinite word $w$ counts, for each integer $n$, the number of distinct fair words of length $n$ occurring in $w$. 
Theorem~\ref{T_balanced} implies that the fair complexity of a Sturmian word is its palindrome complexity.

About studies around infinite words and fairness, let us observe that, in \cite{Cerny2007DAM,Salomaa2007IJFCS},
A.\cerny\ and A.~Salomaa studied another aspect related to fair words in infinite words. They prove (only in the binary case in \cite{Salomaa2007IJFCS}) that: ``\textit{if the first $k+1$ words in the sequence generated by a D0L system over a $k$-letter  alphabet are fair then all words in the sequence are fair}".

To end with questions, perhaps the main problem about 
fair words which stays open is the original {\cerny}'s question \cite{Cerny2009JALC}: count the number of fair words over an arbitrary alphabet.

\bigskip

\noindent
\textbf{Acknowledgements.}
The author would like to thanks Michel Rigo for his invitation to the ``Workshop on words and subwords" he organized at Liège in May 2025 (12-16th): this invitation
encouraged the author starting the study of this paper 
whose results were partially presented at the workshop.

\addcontentsline{toc}{section}{Bibliography}
\bibliographystyle{plainnat} 
\bibliography{around_fair_words}

\pagebreak

\section*{Appendix}
\addcontentsline{toc}{section}{Appendix}

\subsection*{A. Computing the number of occurrences of \texorpdfstring{$ab$}{ab}}
\addcontentsline{toc}{subsection}{A. Computing the number of occurrences of \texorpdfstring{$ab$}{ab}}

The next pictures illustrate the use of Equation~\eqref{eq_calcul_binom_ab_par_prefixes} ($
\binom{w}{ab} = \sum_{pb \text{~prefix of~} w} |p|_a$) to compute $\binom{w}{ab}$.

\begin{center}
\input{grille_6x7_calcul_binom_ab_par_lignes.tex}
\end{center}

\pagebreak

\subsection*{B. Partitions associated with words}
\addcontentsline{toc}{subsection}{B. Partitions associated with words}

The next figures shows the partitions associated with words $w$ such that $2 \leq \binom{w}{ab} \leq 5$.

\medskip

\input{partition_surface2.tex}

\input{partition_surface3.tex}

\input{partition_surface4.tex}

\input{partition_surface5.tex}

\end{document}

%% file: grille_6x7.tex
\begin{figure}[!ht]
\begin{center}
\begin{tikzpicture}[scale=0.8]
\begin{scope}
\draw [dashed] (0,0) grid (7,6);
\draw [red] (0,0) -- (1,0) node[midway,below]{a};
\draw [red] (1,0) -- (2,0) node[midway,below]{a};
\draw [blue] (2,0) -- (2,1) node[midway,right]{b};
\draw [red] (2,1) -- (3,1) node[midway,below]{a};
\draw [red] (3,1) -- (4,1) node[midway,below]{a};
\draw [blue] (4,1) -- (4,2) node[midway,right]{b};
\draw [blue] (4,2) -- (4,3) node[midway,right]{b};
\draw [red] (4,3) -- (5,3) node[midway,below]{a};
\draw [blue] (5,3) -- (5,4) node[midway,right]{b};
\draw [red] (5,4) -- (6,4) node[midway,below]{a};
\draw [blue] (6,4) -- (6,5) node[midway,right]{b};
\draw [blue] (6,5) -- (6,6) node[midway,right]{b};
\draw [red] (6,6) -- (7,6) node[midway,below]{a};
\end{scope}
\end{tikzpicture}
\end{center}
\caption{\label{fig_aabaabbababba}word \ca\ca\cb\ca\ca\cb\cb\ca\cb\ca\cb\cb\ca}
\end{figure}

%% file: grille_6x7_area_formule_base_avec_longueurs.tex
\begin{figure}[!ht]
\begin{center}
\begin{tikzpicture}[scale=0.8]
\begin{scope}
\draw [<->] (-1,0) -- (-1,6) node[midway,left]{height $|w|_b$};
\draw [<->] (0,-1) -- (7, -1) node[midway,below]{width $|w|_a$};

\draw (0,0) -- (0,6);
\draw (0,6) -- (7,6);
\draw (7,6) -- (7,0);
\draw (7,0) -- (0,0);
\draw (2, 5) node[below right] {\LARGE $\binom{w}{ab}$};
\draw (5, 2) node[below right] {\LARGE $\binom{w}{ba}$};
\draw [red] (0,0) -- (1,0) node[midway,below]{a};
\draw [red] (1,0) -- (2,0) node[midway,below]{a};
\draw [blue] (2,0) -- (2,1) node[midway,right]{b};
\draw [red] (2,1) -- (3,1) node[midway,below]{a};
\draw [red] (3,1) -- (4,1) node[midway,below]{a};
\draw [blue] (4,1) -- (4,2) node[midway,right]{b};
\draw [blue] (4,2) -- (4,3) node[midway,right]{b};
\draw [red] (4,3) -- (5,3) node[midway,below]{a};
\draw [blue] (5,3) -- (5,4) node[midway,right]{b};
\draw [red] (5,4) -- (6,4) node[midway,below]{a};
\draw [blue] (6,4) -- (6,5) node[midway,right]{b};
\draw [blue] (6,5) -- (6,6) node[midway,right]{b};
\draw [red] (6,6) -- (7,6) node[midway,below]{a};
\end{scope}
\end{tikzpicture}
\end{center}
\caption{\label{fig_area}$\binom{w}{ab}$ and $\binom{w}{ba}$ for word $w =$\ca\ca\cb\ca\ca\cb\cb\ca\cb\ca\cb\cb\ca}
\end{figure}
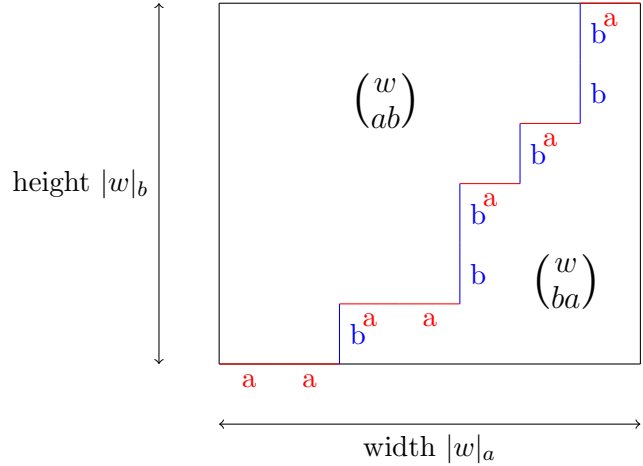

%% file: interpretation_sommes.tex
\begin{figure}[!ht]
\begin{center}
\begin{tikzpicture}[scale=0.8]
\begin{scope}
\draw [<->] (-1,0) -- (-1,6) node[midway,left]{height $|w|_b$};
\draw [<->] (0,-1) -- (13, -1) node[midway,below]{width $|w|$};

\draw (0,0) -- (0,6);
\draw (0,6) -- (13,6);
\draw (13,6) -- (13,0);
\draw (13,0) -- (0,0);
\draw (2, 5) node[below right] {\LARGE $S_b(w)-\frac{|w|_b}{2}$};
\draw (8, 2) node[below right] {\LARGE $S_b'(w)-\frac{|w|_b}{2}$};
\draw [red] (0,0) -- (1,0) node[midway,below]{a};
\draw [red] (1,0) -- (2,0) node[midway,below]{a};
\draw [blue] (2,0) -- (3,1) node[midway,right]{b};
\draw [red] (3,1) -- (4,1) node[midway,below]{a};
\draw [red] (4,1) -- (5,1) node[midway,below]{a};
\draw [blue] (5,1) -- (6,2) node[midway,right]{b};
\draw [blue] (6,2) -- (7,3) node[midway,right]{b};
\draw [red] (7,3) -- (8,3) node[midway,below]{a};
\draw [blue] (8,3) -- (9,4) node[midway,right]{b};
\draw [red] (9,4) -- (10,4) node[midway,below]{a};
\draw [blue] (10,4) -- (11,5) node[midway,right]{b};
\draw [blue] (11,5) -- (12,6) node[midway,right]{b};
\draw [red] (12,6) -- (13,6) node[midway,below]{a};
\end{scope}
\end{tikzpicture}
\end{center}
\caption{\label{fig_area_sums}Geometrical interpretations of $S_b(w)$ and $S_b'(w)$}
\end{figure}
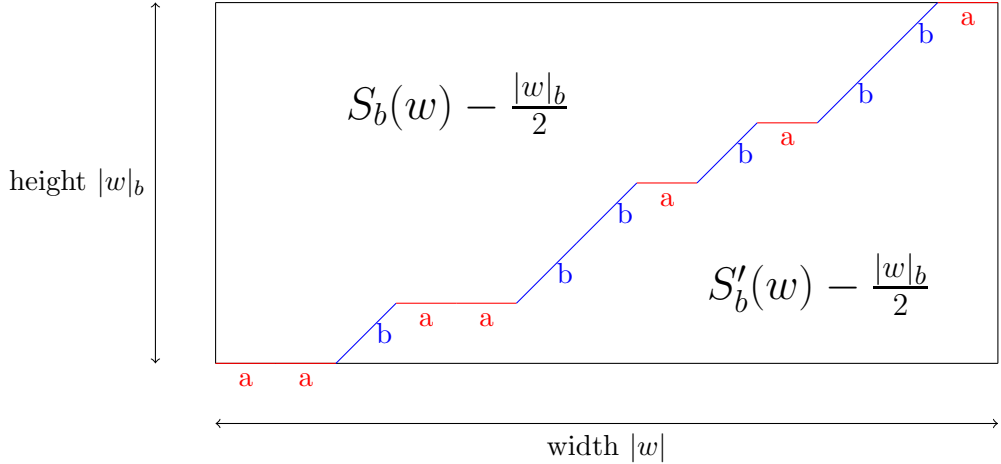

%% file: congruence.tex
\begin{figure}[!ht]
\begin{center}
\begin{tikzpicture}[scale=0.8]
\begin{scope}
\draw (0,0) -- (0,6);
\draw (0,6) -- (7,6);
\draw (7,6) -- (7,0);
\draw (7,0) -- (0,0);
\draw (2.75, 3.5) node[below right] {$\binom{u}{ab}$};
\draw (-0.1, 1) node[below right] {$\binom{x}{ab}$};
\draw (5.9, 6) node[below right] {$\binom{y}{ab}$};
\draw (0, 5.5) node[below right] {$|x|_a|y|_b$};
\draw (0, 3) node[below right] {$|x|_a|u|_b$};
\draw (3, 5.5) node[below right] {$|u|_a|y|_b$};
\draw (0,0) -- (2,1) node[midway,below]{$x$};
\fill [lightgray,opacity=0.5] (0,0) -- (2,1) -- (0,1) -- cycle ;
\fill [lightgray,opacity=0.5] (0,1) -- (2,1) -- (2,4) -- (0, 4) -- cycle ;
\fill [lightgray,opacity=0.5] (0,4) -- (2,4) -- (2,6) -- (0, 6) -- cycle ;
\fill [lightgray,opacity=0.5] (2,4) -- (6,4) -- (6,6) -- (2, 6) -- cycle ;
\fill [lightgray,opacity=0.5] (2,1) -- (2,4) -- (6,4) -- cycle ;
\fill [lightgray,opacity=0.5] (6,4) -- (6,6) -- (7,6) -- cycle ;
\draw [dashed] (0,1) -- (7,1) ;
\draw [dashed] (0,4) -- (7,4) ;
\draw [dashed] (2,0) -- (2,6) ;
\draw [dashed] (6,0) -- (6,6) ;
\draw [<->] (-1,0) -- (-1,1) node[midway,left]{$|x|_b$};
\draw [<->] (-1,1) -- (-1,4) node[midway,left]{$|u|_b$};
\draw [<->] (-1,4) -- (-1,6) node[midway,left]{$|y|_b$};
\draw [<->] (0,7) -- (2,7) node[midway,below]{$|x|_a$};
\draw [<->] (2,7) -- (6,7) node[midway,below]{$|u|_a$};
\draw [<->] (6,7) -- (7,7) node[midway,below]{$|x|_a$};
\draw [blue] (2,1) -- (6,4) node[midway,right]{$u$};
\draw (6,4) -- (7,6) node[midway,below]{$y$};
\end{scope}
\end{tikzpicture}
\end{center}
\caption{\label{figure_congruence}Number of occurrences of $ab$ in the concatenation of three words}
\end{figure}
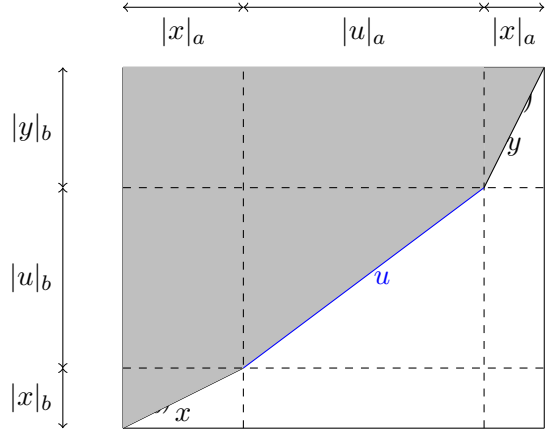

%% file: passage_xaby_xbay.tex
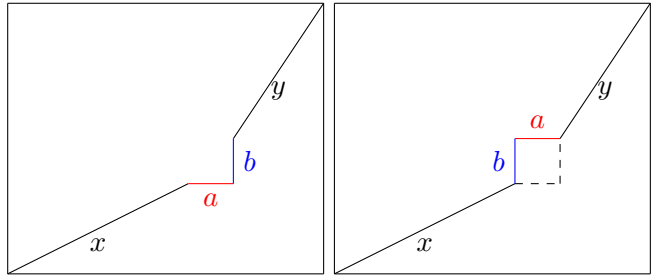
\begin{figure}[!ht]
\begin{center}
\begin{tikzpicture}[scale=0.6]
\begin{scope}
\draw (0,0) -- (0,6);
\draw (0,6) -- (7,6);
\draw (7,6) -- (7,0);
\draw (7,0) -- (0,0);
\draw (0,0) -- (4,2)  node[midway,below]{$x$};
\draw [red] (4,2) -- (5,2) node[midway,below]{$a$};
\draw [blue] (5,2) -- (5,3) node[midway,right]{$b$};
\draw (5,3) -- (7,6)  node[midway,below]{$y$};
\end{scope}
\end{tikzpicture}
\begin{tikzpicture}[scale=0.6]
\begin{scope}
\draw (0,0) -- (0,6);
\draw (0,6) -- (7,6);
\draw (7,6) -- (7,0);
\draw (7,0) -- (0,0);
\draw (0,0) -- (4,2)  node[midway,below]{$x$};
\draw [blue] (4,2) -- (4,3) node[midway,left]{$b$};
\draw [red] (4,3) -- (5,3) node[midway,above]{$a$};
\draw [dashed] (4,2) -- (5,2);
\draw [dashed] (5,2) -- (5,3);
\draw (5,3) -- (7,6)  node[midway,below]{$y$};
\end{scope}
\end{tikzpicture}
\end{center}
\caption{\label{fig_interpretation_baisse_aire}From $uabv$ to $ubav$}
\end{figure}

%% file: passage_xabybaz_xbayabz.tex
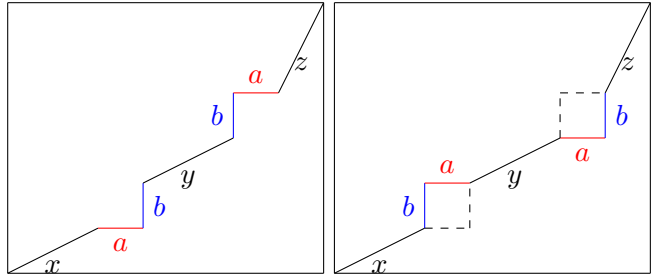
\begin{figure}[!ht]
\begin{center}
\begin{tikzpicture}[scale=0.6]
\begin{scope}
\draw (0,0) -- (0,6);
\draw (0,6) -- (7,6);
\draw (7,6) -- (7,0);
\draw (7,0) -- (0,0);
\draw (0,0) -- (2,1)  node[midway,below]{$x$};
\draw [red] (2,1) -- (3,1) node[midway,below]{$a$};
\draw [blue] (3,1) -- (3,2) node[midway,right]{$b$};
\draw (3,2) -- (5,3)  node[midway,below]{$y$};
\draw [blue] (5,3) -- (5,4) node[midway,left]{$b$};
\draw [red] (5,4) -- (6,4) node[midway,above]{$a$};
\draw (6,4) -- (7,6)  node[midway,below]{$z$};
\end{scope}
\end{tikzpicture}
\begin{tikzpicture}[scale=0.6]
\begin{scope}
\draw (0,0) -- (0,6);
\draw (0,6) -- (7,6);
\draw (7,6) -- (7,0);
\draw (7,0) -- (0,0);
\draw (0,0) -- (2,1)  node[midway,below]{$x$};
\draw [blue] (2,1) -- (2,2) node[midway,left]{$b$};
\draw [red] (2,2) -- (3,2) node[midway,above]{$a$};
\draw [dashed] (2,1) -- (3,1) ;
\draw [dashed] (3,1) -- (3,2) ;
\draw (3,2) -- (5,3)  node[midway,below]{$y$};
\draw [red] (5,3) -- (6,3) node[midway,below]{$a$};
\draw [blue] (6,3) -- (6,4) node[midway,right]{$b$};
\draw [dashed] (5,3) -- (5,4) ;
\draw [dashed] (5,4) -- (6,4) ;
\draw (6,4) -- (7,6)  node[midway,below]{$z$};

\end{scope}
\end{tikzpicture}
\end{center}
\caption{\label{fig_passage_xabybaz_xbayabz}From $xabybaz$ to $xbayabz$}
\end{figure}

%% file: about_2_binomial_equivalence.tex
\begin{figure}[!ht]
\begin{center}
\begin{tikzpicture}[scale=0.4]
\begin{scope}
\draw [dotted] (0,0) -- (16,0);
\draw [dotted] (16,0) -- (16,11);
\draw [dotted] (16,11) -- (0,11);
\draw [dotted] (0,0) -- (0,11);
\draw [dotted]  grid (16,11);
\draw [red] (0,0) -- (1,0) ;
\draw [red] (1,0) -- (2,0) ;
\draw [red] (2,0) -- (2,1) ;
\draw [red] (2,1) -- (3,1) ;
\fill [gray!20,opacity=0.5] (0,0) -- (2,0) -- (2,1) -- (0,1) -- cycle ;
\fill [gray!20,opacity=0.5] (1,1) -- (3,1) -- (3,2) -- (1,2) -- cycle ;
\fill [gray!20,opacity=0.5] (4,3) -- (5,3) -- (5,5) -- (4,5) -- cycle ;
\fill [gray!20,opacity=0.5] (5,4) -- (6,4) -- (6,7) -- (5,7) -- cycle ;
\fill [gray!20,opacity=0.5] (7,7) -- (9,7) -- (9,8) -- (7,8) -- cycle ;
\fill [gray!20,opacity=0.5] (9,8) -- (14,8) -- (14,9) -- (9,9) -- cycle ;
\fill [gray!20,opacity=0.5] (10,9) -- (14,9) -- (14,10) -- (10,10) -- cycle ;
\fill [gray!20,opacity=0.5] (14,10) -- (16,10) -- (16,11) -- (14,11) -- cycle ;
\draw [red] (3,1) -- (3,2) ;
\draw [red] (3,2) -- (3,3) ;
\draw [red] (3,3) -- (4,3) ;
\draw [dashed] (0,0) -- (0,1) ;
\draw [dashed] (0,1) -- (1,1) ;
\draw [dashed] (1,1) -- (1,2) ;
\draw [dashed] (1,2) -- (2,2) ;
\draw [dashed] (2,2) -- (3,2) ;
\draw [dashed] (3,2) -- (3,3) ;
\draw [dashed] (3,3) -- (4,3) ;
\draw [red] (4,3) -- (5,3) ;
\draw [red] (5,3) -- (5,4) ;
\draw [red] (5,4) -- (6,4) ;
\draw [red] (6,4) -- (6,5) ;
\draw [red] (6,5) -- (6,6) ;
\draw [red] (6,6) -- (6,7) ;
\draw [dashed](4,3) -- (4,4) ;
\draw [dashed](4,4) -- (4,5) ;
\draw [dashed](4,5) -- (5,5) ;
\draw [dashed](5,5) -- (5,6) ;
\draw [dashed](5,6) -- (5,7) ;
\draw [dashed](5,7) -- (6,7) ;
\draw [red] (7,8) -- (8,8) ;
\draw [red] (7,7) -- (7,8) ;
\draw [red] (8,8) -- (9,8) ;
\draw [dashed] (6,7) -- (7,7) ;
\draw [red] (6,7) -- (7,7) ;
\draw [dashed] (7,7) -- (8,7) ;
\draw [dashed] (8,7) -- (9,7) ;
\draw [dashed] (9,7) -- (9,8) ;
\draw [red] (9,8) -- (9,9) ;
\draw [red] (9,9) -- (10,9) ;
\draw [red] (10,9) -- (10,10) ;
\draw [red] (10,10) -- (11,10) ;
\draw [red] (11,10) -- (12,10) ;
\draw [red] (12,10) -- (13,10) ;
\draw [red] (13,10) -- (14,10) ;
\draw [dashed] (9,8) -- (10,8) ;
\draw [dashed] (10,8) -- (11,8) ;
\draw [dashed] (11,8) -- (12,8) ;
\draw [dashed] (12,8) -- (13,8) ;
\draw [dashed] (13,8) -- (14,8) ;
\draw [dashed] (14,8) -- (14,9) ;
\draw [dashed] (14,9) -- (14,10);
\draw [red] (14,10) -- (15,10) ;
\draw [red] (15,10) -- (16,10) ;
\draw [red] (16,10) -- (16,11) ;
\draw [dashed] (14,10) -- (14,11);
\draw [dashed] (14,11) -- (15,11);
\draw [dashed] (15,11) -- (16,11);
\end{scope}
\end{tikzpicture}
\end{center}
\caption{\label{fig_2_words}two $2$-binomial equivalent words}
\end{figure}
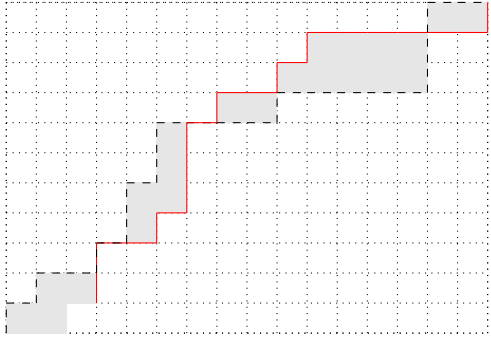

%% file: exemple_graphes_classe.tex
\begin{figure}[!ht]
\begin{center}
\begin{tikzpicture}[scale=0.3]
\begin{scope}
\node[draw] (a) at (0,4) {$ab^5a^4$};
\node[draw] (b) at (8,4) {$bab^3aba^3$};
\node[draw] (d) at (16,4) {$b^2ab^2a^2ba^2$};
\node[draw] (f) at (24,4) {$b^3aba^3ba$};
\node[draw] (g) at (32,4) {$b^4a^5b$};

\node[draw] (c) at (12,0) {$b^2abab^2a^3$};
\node[draw] (e) at (20,0) {$b^3a^2baba^2$};

\draw [->] (a) -- (b);
\draw [->] (b) -- (c);
\draw [->] (b) -- (d);
\draw [->] (c) -- (d);
\draw [->] (c) -- (e);
\draw [->] (d) -- (e);
\draw [->] (d) -- (f);
\draw [->] (e) -- (f);
\draw [->] (f) -- (g);

\end{scope}
\end{tikzpicture}
\end{center}
\caption{\label{fig_lattice_example}The graph 
$([ab^5a^4]_{\sim_2}, \rightarrow)$}
\end{figure}
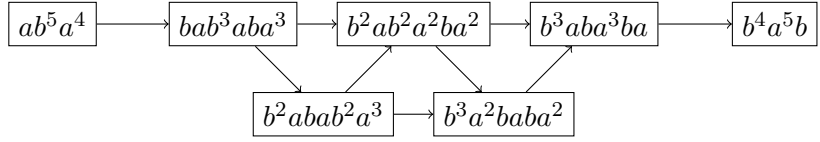

%% file: exemple_init_final.tex
\begin{figure}[!ht]
\begin{center}
\begin{tikzpicture}[scale=0.7]
\begin{scope}
\draw [white] (-1,-1) -- (-1,7);
\draw [white] (-1,7) -- (8,7);
\draw [white] (8,7) -- (8,-1);
\draw [white] (8,-1) -- (-1,-1);
\draw [dashed] (0,0) -- (0,6);
\draw [dashed] (0,6) -- (7,6);
\draw [dashed] (7,6) -- (7,0);
\draw [dashed]  (7,0) -- (0,0);
\draw [dashed] (0,0) grid (4,6) ;
\draw [dashed] (4,4) grid (5,6) ;
\draw [<->] (0,7) -- (4,7) node[midway,below]{$i$};
\draw [<->] (-1,3) -- (-1,6) node[midway,left]{$j$};
\draw [red] (0,0) -- (1,0) node[midway,below]{a};
\draw [red] (1,0) -- (2,0) node[midway,below]{a};
\draw [red] (2,0) -- (3,0) node[midway,below]{a};
\draw [red] (3,0) -- (4,0) node[midway,below]{a};
\draw [blue] (4,0) -- (4,1) node[midway,right]{b};
\draw [blue] (4,1) -- (4,2) node[midway,right]{b};
\draw [blue] (4,2) -- (4,3) node[midway,right]{b};
\draw [red] (4,3) -- (5,3) node[midway,below]{a};
\draw [blue] (5,3) -- (5,4) node[midway,right]{b};
\draw [blue] (5,4) -- (5,5) node[midway,right]{b};
\draw [blue] (5,5) -- (5,6) node[midway,right]{b};
\draw [red] (5,6) -- (6,6) node[midway,below]{a};
\draw [red] (6,6) -- (7,6) node[midway,below]{a};
\end{scope}
\end{tikzpicture}
~~
\begin{tikzpicture}[scale=0.7]
\begin{scope}
\draw [white] (-1,-1) -- (-1,7);
\draw [white] (-1,7) -- (8,7);
\draw [white] (8,7) -- (8,-1);
\draw [white] (8,-1) -- (-1,-1);
\draw [dashed] (0,0) -- (0,6);
\draw [dashed] (0,6) -- (7,6);
\draw [dashed] (7,6) -- (7,0);
\draw [dashed]  (7,0) -- (0,0);
\draw [dashed] (0,2) grid (6,6) ;
\draw [dashed] (6,3) grid (7,6) ;
\draw [<->] (-1,3) -- (-1,6) node[midway,left]{$i$};
\draw [<->] (0,7) -- (6,7) node[midway,below]{$j$};
\draw [blue] (0,0) -- (0,1) node[midway,right]{b};
\draw [blue] (0,1) -- (0,2) node[midway,right]{b};%
\draw [red] (0,2) -- (1,2) node[midway,below]{a};%
\draw [red] (1,2) -- (2,2) node[midway,below]{a};
\draw [red] (2,2) -- (3,2) node[midway,below]{a};
\draw [red] (3,2) -- (4,2) node[midway,below]{a};
\draw [red] (4,2) -- (5,2) node[midway,below]{a};
\draw [red] (5,2) -- (6,2) node[midway,below]{a};%
\draw [blue] (6,2) -- (6,3) node[midway,right]{b};%
\draw [red] (6,3) -- (7,3) node[midway,below]{a};
\draw [blue] (7,3) -- (7,4) node[midway,right]{b};
\draw [blue] (7,4) -- (7,5) node[midway,right]{b};
\draw [blue] (7,5) -- (7,6) node[midway,right]{b};
\end{scope}
\end{tikzpicture}
\end{center}
\caption{\label{fig_exemple_init_and_final} $\init(aabaabbababba)$ and $\final(aabaabbababba)$}
\end{figure}
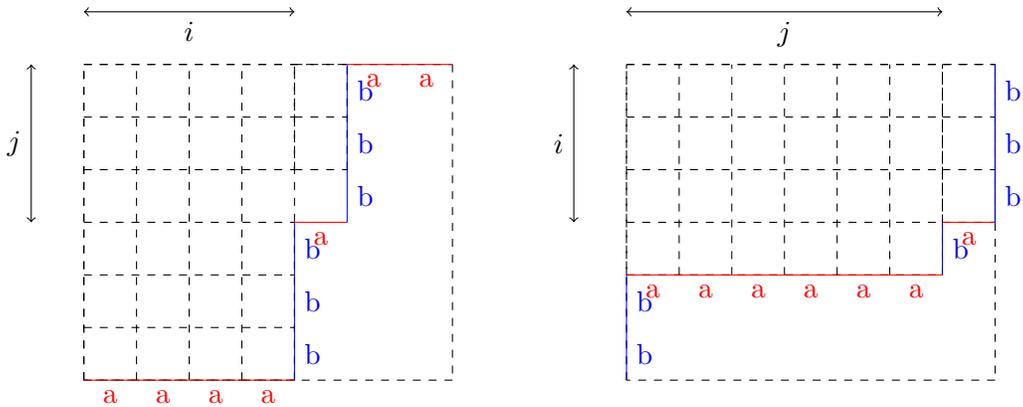

%% file: singletons.tex
\begin{figure}[!ht]
\begin{center}
\begin{tikzpicture}[scale=0.2]
\begin{scope}
\draw [blue] (0,0) -- (0,6) node[midway,right]{$b^*$};
\draw [red] (0,6) -- (7,6) node[midway,below]{$a^*$};
\end{scope}
\end{tikzpicture}
~~~
\begin{tikzpicture}[scale=0.2]
\begin{scope}
\draw [blue] (0,0) -- (0,3) node[near start,right]{$b^*$};
\draw [red] (0,3) -- (2,3) node[midway,below]{$a$};
\draw [blue] (2,3) -- (2,6) node[midway,right]{$b^*$};
\end{scope}
\end{tikzpicture}
~~~
\begin{tikzpicture}[scale=0.2]
\begin{scope}
\draw [blue] (0,0) -- (0,4) node[near start,right]{$b^*$};
\draw [red] (0,4) -- (2,4) node[midway,below]{$a$};
\draw [blue] (2,4) -- (2,6) node[midway,right]{$b$};
\draw [red] (2,6) -- (7,6) node[near end,below]{$a^*$};
\end{scope}
\end{tikzpicture}
~~~
\begin{tikzpicture}[scale=0.2]
\begin{scope}
\draw [red] (0,0) -- (7,0) node[midway,above]{$a^*$};
\draw [blue] (7,0) -- (7,6) node[midway,right]{$b^*$};
\end{scope}
\end{tikzpicture}
~~~
\begin{tikzpicture}[scale=0.2]
\begin{scope}
\draw [red] (0,0) -- (4,0) node[near start,above]{$a^*$};
\draw [blue] (4,0) -- (4,2) node[midway,left]{$b$};
\draw [red] (4,2) -- (7,2) node[midway,above]{$a^*$};
\end{scope}
\end{tikzpicture}
~~~
\begin{tikzpicture}[scale=0.2]
\begin{scope}
\draw [red] (0,0) -- (6,0) node[midway,above]{$a^*$};
\draw [blue] (6,0) -- (6,2) node[midway,left]{b};
\draw [red] (6,2) -- (7,2) node[midway,above]{a};
\draw [blue] (7,2) -- (7,6) node[near end,left]{$b^*$};
\end{scope}
\end{tikzpicture}
\end{center}
\caption{\label{fig_singletons}Shapes of words $w$ with $\#[w]_{\sim_2} = 1$}
\end{figure}
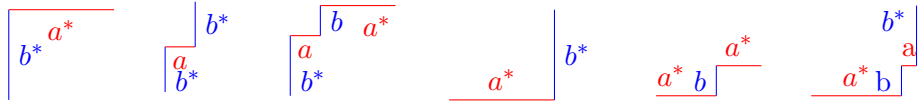

%% file: grille_6x7_exemple_illustration_fair_article.tex
\begin{figure}[!ht]
\begin{center}
\begin{tikzpicture}[scale=0.8]
\begin{scope}
\draw (0,0) -- (0,6);
\draw (0,6) -- (7,6);
\draw (7,6) -- (7,0);
\draw (7,0) -- (0,0);
\draw (2, 5) node[below right] {\LARGE $\binom{w}{ab}$};
\draw (3, 2) node[below right] {\LARGE $\binom{w}{ba}=\binom{w}{ab}$};
\draw [red] (0,0) -- (1,0) node[midway,above]{a};
\draw [blue] (1,0) -- (1,1) node[midway,right]{b};
\draw [blue] (1,1) -- (1,2) node[midway,right]{b};
\draw [red] (1,2) -- (2,2) node[midway,below]{a};
\draw [red] (2,2) -- (3,2) node[midway,below]{a};
\draw [blue] (3,2) -- (3,3) node[midway,right]{b};
\draw [red] (3,3) -- (4,3) node[midway,below]{a};
\draw [red] (4,3) -- (5,3) node[midway,below]{a};
\draw [blue] (5,3) -- (5,4) node[midway,right]{b};
\draw [blue] (5,4) -- (5,5) node[midway,right]{b};
\draw [red] (5,5) -- (6,5) node[midway,below]{a};
\draw [blue] (6,5) -- (6,6) node[midway,right]{b};
\draw [red] (6,6) -- (7,6) node[midway,below]{a};\end{scope}
\end{tikzpicture}
\end{center}
\caption{\label{fig_fair_word_example}Example of fair word}
\end{figure}

%% file: grille_6x7_calcul_binom_ab_par_lignes.tex
\begin{figure}[!ht]
\begin{tikzpicture}[scale=0.7]
\begin{scope}
\draw (0,0) -- (0,6);
\draw (0,6) -- (7,6);
\draw (7,6) -- (7,0);
\draw (7,0) -- (0,0);
\draw [dashed] (0,0) grid (2,1);
\draw [red] (0,0) -- (1,0) node[midway,below]{a};
\draw [red] (1,0) -- (2,0) node[midway,below]{a};
\draw [blue] (2,0) -- (2,1) node[midway,right]{b};
\draw [red] (2,1) -- (3,1) node[midway,below]{a};
\draw [red] (3,1) -- (4,1) node[midway,below]{a};
\draw [blue] (4,1) -- (4,2) node[midway,right]{b};
\draw [blue] (4,2) -- (4,3) node[midway,right]{b};
\draw [red] (4,3) -- (5,3) node[midway,below]{a};
\draw [blue] (5,3) -- (5,4) node[midway,right]{b};
\draw [red] (5,4) -- (6,4) node[midway,below]{a};
\draw [blue] (6,4) -- (6,5) node[midway,right]{b};
\draw [blue] (6,5) -- (6,6) node[midway,right]{b};
\draw [red] (6,6) -- (7,6) node[midway,below]{a};
\end{scope}
\end{tikzpicture}
~~
\begin{tikzpicture}[scale=0.7]
\begin{scope}
\draw (0,0) -- (0,6);
\draw (0,6) -- (7,6);
\draw (7,6) -- (7,0);
\draw (7,0) -- (0,0);
\draw [dashed] (0,0) grid (2,1);
\draw [dashed] (0,1) grid (4,2);
\draw [red] (0,0) -- (1,0) node[midway,below]{a};
\draw [red] (1,0) -- (2,0) node[midway,below]{a};
\draw [blue] (2,0) -- (2,1) node[midway,right]{b};
\draw [red] (2,1) -- (3,1) node[midway,below]{a};
\draw [red] (3,1) -- (4,1) node[midway,below]{a};
\draw [blue] (4,1) -- (4,2) node[midway,right]{b};
\draw [blue] (4,2) -- (4,3) node[midway,right]{b};
\draw [red] (4,3) -- (5,3) node[midway,below]{a};
\draw [blue] (5,3) -- (5,4) node[midway,right]{b};
\draw [red] (5,4) -- (6,4) node[midway,below]{a};
\draw [blue] (6,4) -- (6,5) node[midway,right]{b};
\draw [blue] (6,5) -- (6,6) node[midway,right]{b};
\draw [red] (6,6) -- (7,6) node[midway,below]{a};
\end{scope}
\end{tikzpicture}
~~
\begin{tikzpicture}[scale=0.7]
\begin{scope}
\draw (0,0) -- (0,6);
\draw (0,6) -- (7,6);
\draw (7,6) -- (7,0);
\draw (7,0) -- (0,0);
\draw [dashed] (0,0) grid (2,1);
\draw [dashed] (0,1) grid (4,2);
\draw [dashed] (0,2) grid (4,3);
\draw [red] (0,0) -- (1,0) node[midway,below]{a};
\draw [red] (1,0) -- (2,0) node[midway,below]{a};
\draw [blue] (2,0) -- (2,1) node[midway,right]{b};
\draw [red] (2,1) -- (3,1) node[midway,below]{a};
\draw [red] (3,1) -- (4,1) node[midway,below]{a};
\draw [blue] (4,1) -- (4,2) node[midway,right]{b};
\draw [blue] (4,2) -- (4,3) node[midway,right]{b};
\draw [red] (4,3) -- (5,3) node[midway,below]{a};
\draw [blue] (5,3) -- (5,4) node[midway,right]{b};
\draw [red] (5,4) -- (6,4) node[midway,below]{a};
\draw [blue] (6,4) -- (6,5) node[midway,right]{b};
\draw [blue] (6,5) -- (6,6) node[midway,right]{b};
\draw [red] (6,6) -- (7,6) node[midway,below]{a};
\end{scope}
\end{tikzpicture}
\caption{\label{comput_binom_ab_etapes_1_a_3}Three first steps computing $\binom{w}{ \ca\ca\cb\ca\ca\cb\cb\ca\cb\ca\cb\cb\ca}$}
\end{figure}

\begin{figure}[!ht]
\begin{tikzpicture}[scale=0.7]
\begin{scope}
\draw (0,0) -- (0,6);
\draw (0,6) -- (7,6);
\draw (7,6) -- (7,0);
\draw (7,0) -- (0,0);
\draw [dashed] (0,0) grid (2,1);
\draw [dashed] (0,1) grid (4,2);
\draw [dashed] (0,2) grid (4,3);
\draw [dashed] (0,3) grid (5,4);
\draw [red] (0,0) -- (1,0) node[midway,below]{a};
\draw [red] (1,0) -- (2,0) node[midway,below]{a};
\draw [blue] (2,0) -- (2,1) node[midway,right]{b};
\draw [red] (2,1) -- (3,1) node[midway,below]{a};
\draw [red] (3,1) -- (4,1) node[midway,below]{a};
\draw [blue] (4,1) -- (4,2) node[midway,right]{b};
\draw [blue] (4,2) -- (4,3) node[midway,right]{b};
\draw [red] (4,3) -- (5,3) node[midway,below]{a};
\draw [blue] (5,3) -- (5,4) node[midway,right]{b};
\draw [red] (5,4) -- (6,4) node[midway,below]{a};
\draw [blue] (6,4) -- (6,5) node[midway,right]{b};
\draw [blue] (6,5) -- (6,6) node[midway,right]{b};
\draw [red] (6,6) -- (7,6) node[midway,below]{a};
\end{scope}
\end{tikzpicture}
~~
\begin{tikzpicture}[scale=0.7]
\begin{scope}
\draw (0,0) -- (0,6);
\draw (0,6) -- (7,6);
\draw (7,6) -- (7,0);
\draw (7,0) -- (0,0);
\draw [dashed] (0,0) grid (2,1);
\draw [dashed] (0,1) grid (4,2);
\draw [dashed] (0,2) grid (4,3);
\draw [dashed] (0,3) grid (5,4);
\draw [dashed] (0,4) grid (6,5);
\draw [red] (0,0) -- (1,0) node[midway,below]{a};
\draw [red] (1,0) -- (2,0) node[midway,below]{a};
\draw [blue] (2,0) -- (2,1) node[midway,right]{b};
\draw [red] (2,1) -- (3,1) node[midway,below]{a};
\draw [red] (3,1) -- (4,1) node[midway,below]{a};
\draw [blue] (4,1) -- (4,2) node[midway,right]{b};
\draw [blue] (4,2) -- (4,3) node[midway,right]{b};
\draw [red] (4,3) -- (5,3) node[midway,below]{a};
\draw [blue] (5,3) -- (5,4) node[midway,right]{b};
\draw [red] (5,4) -- (6,4) node[midway,below]{a};
\draw [blue] (6,4) -- (6,5) node[midway,right]{b};
\draw [blue] (6,5) -- (6,6) node[midway,right]{b};
\draw [red] (6,6) -- (7,6) node[midway,below]{a};
\end{scope}
\end{tikzpicture}
~~
\begin{tikzpicture}[scale=0.7]
\begin{scope}
\draw (0,0) -- (0,6);
\draw (0,6) -- (7,6);
\draw (7,6) -- (7,0);
\draw (7,0) -- (0,0);
\draw [dashed] (0,0) grid (2,1);
\draw [dashed] (0,1) grid (4,2);
\draw [dashed] (0,2) grid (4,3);
\draw [dashed] (0,3) grid (5,4);
\draw [dashed] (0,4) grid (6,5);
\draw [dashed] (0,5) grid (6,6);
\draw [red] (0,0) -- (1,0) node[midway,below]{a};
\draw [red] (1,0) -- (2,0) node[midway,below]{a};
\draw [blue] (2,0) -- (2,1) node[midway,right]{b};
\draw [red] (2,1) -- (3,1) node[midway,below]{a};
\draw [red] (3,1) -- (4,1) node[midway,below]{a};
\draw [blue] (4,1) -- (4,2) node[midway,right]{b};
\draw [blue] (4,2) -- (4,3) node[midway,right]{b};
\draw [red] (4,3) -- (5,3) node[midway,below]{a};
\draw [blue] (5,3) -- (5,4) node[midway,right]{b};
\draw [red] (5,4) -- (6,4) node[midway,below]{a};
\draw [blue] (6,4) -- (6,5) node[midway,right]{b};
\draw [blue] (6,5) -- (6,6) node[midway,right]{b};
\draw [red] (6,6) -- (7,6) node[midway,below]{a};
\end{scope}
\end{tikzpicture}

\caption{\label{fig_comput_binom_ab_etapes_4_a_6}Three last steps computing $\binom{w}{ \ca\ca\cb\ca\ca\cb\cb\ca\cb\ca\cb\cb\ca}$}
\end{figure}

%% file: partition_surface2.tex
\begin{figure}[!ht]
\begin{center}
\begin{tikzpicture}[scale=0.4]
\begin{scope}
\draw [blue] (0,0) -- (0,4) node[midway,right]{$b^*$};
\draw [red] (1,6) -- (7,6) node[midway,below]{$a^*$};
\draw [dashed] (0,4) grid (1,6) ;
\draw [red] (0,4) -- (1,4) node[midway,below]{a};
\draw [blue] (1,4) -- (1,5) node[midway,right]{b};
\draw [blue] (1,5) -- (1,6) node[midway,right]{b};
\end{scope}
\end{tikzpicture}
~~~
\begin{tikzpicture}[scale=0.4]
\begin{scope}
\draw [blue] (0,0) -- (0,5) node[midway,right]{$b^*$};
\draw [red] (2,6) -- (7,6) node[midway,below]{$a^*$};
\draw [dashed] (0,5) grid (2,6) ;
\draw [red] (0,5) -- (1,5) node[midway,below]{a};
\draw [red] (1,5) -- (2,5) node[midway,below]{a};
\draw [blue] (2,5) -- (2,6) node[midway,right]{b};
\end{scope}
\end{tikzpicture}
\end{center}
\caption{\label{fig_card_2}Partitions when $\binom{w}{ab} = 2$}
\end{figure}

%% file: partition_surface3.tex
\begin{figure}[!ht]
\begin{center}
\begin{tikzpicture}[scale=0.4]
\begin{scope}
\draw [blue] (0,0) -- (0,3) node[midway,right]{$b^*$};
\draw [red] (1,6) -- (7,6) node[midway,below]{$a^*$};
\draw [dashed] (0,3) grid (1,6) ;
\draw [red] (0,3) -- (1,3) node[midway,below]{a};
\draw [blue] (1,3) -- (1,4) node[midway,right]{b};
\draw [blue] (1,4) -- (1,5) node[midway,right]{b};
\draw [blue] (1,5) -- (1,6) node[midway,right]{b};
\end{scope}
\end{tikzpicture}
~~~
\begin{tikzpicture}[scale=0.4]
\begin{scope}
\draw [blue] (0,0) -- (0,4) node[midway,right]{$b^*$};
\draw [red] (0,4) -- (1,4) node[midway,below]{a};
\draw [blue] (1,4) -- (1,5) node[midway,right]{b};
\draw [red] (1,5) -- (2,5) ;
\draw [red] (1.6,4.35) node[right]{a};
\draw [blue] (2,5) -- (2,6) node[midway,right]{b};
\draw [red] (2,6) -- (7,6) node[midway,below]{$a^*$};
\draw [dashed] (0,4) grid (1,6) ;
\draw [dashed] (1,5) grid (2,6) ;
\end{scope}
\end{tikzpicture}
~~~
\begin{tikzpicture}[scale=0.4]
\begin{scope}
\draw [blue] (0,0) -- (0,5) node[midway,right]{$b^*$};
\draw [red] (3,6) -- (7,6) node[midway,below]{$a^*$};
\draw [dashed] (0,5) grid (3,6) ;
\draw [red] (0,5) -- (1,5) node[midway,below]{a};
\draw [red] (1,5) -- (2,5) node[midway,below]{a};
\draw [red] (2,5) -- (3,5) node[midway,below]{a};
\draw [blue] (3,5) -- (3,6) node[midway,right]{b};
\end{scope}
\end{tikzpicture}
\end{center}
\caption{\label{fig_card_3}Partitions when $\binom{w}{ab} = 3$}

\end{figure}

%% file: partition_surface4.tex
\begin{figure}[!ht]
\begin{center}
\begin{tikzpicture}[scale=0.4]
\begin{scope}
\draw [dashed] (0,0) grid (1,4) ;
\draw [red] (0,0) -- (1,0) ;
\draw [blue] (1,0) -- (1,1) ;
\draw [blue] (1,0) -- (1,2) ;
\draw [blue] (1,0) -- (1,3) ;
\draw [blue] (1,0) -- (1,4) ;
\draw [red] (1,4) -- (4,4) ;
\end{scope}
\end{tikzpicture}
~~~
\begin{tikzpicture}[scale=0.4]
\begin{scope}
\draw [dashed] (0,1) grid (1,4) ;
\draw [dashed] (1,3) grid (2,4) ;
\draw [blue] (0,0) -- (0,1) ;
\draw [red] (0,1) -- (1,1) ;
\draw [blue] (1,1) -- (1,2) ;
\draw [blue] (1,1) -- (1,3) ;
\draw [red] (1,3) -- (2,3) ;
\draw [blue] (2,3) -- (2,4) ;
\draw [red] (2,4) -- (4,4) ;
\end{scope}
\end{tikzpicture}
~~~
\begin{tikzpicture}[scale=0.4]
\begin{scope}
\draw [dashed] (0,2) grid (2,4) ;
\draw [blue] (0,0) -- (0,2) ;
\draw [red] (0,2) -- (2,2) ;
\draw [blue] (2,2) -- (2,4) ;
\draw [red] (2,4) -- (4,4) ;
\end{scope}
\end{tikzpicture}
~~~
\begin{tikzpicture}[scale=0.4]
\begin{scope}
\draw [dashed] (0,2) grid (1,4) ;
\draw [dashed] (1,3) grid (3,4) ;
\draw [blue] (0,0) -- (0,2) ;
\draw [red] (0,2) -- (1,2) ;
\draw [blue] (1,2) -- (1,3) ;
\draw [red] (1,3) -- (3,3) ;
\draw [blue] (3,3) -- (3,4) ;
\draw [red] (3,4) -- (4,4) ;
\end{scope}
\end{tikzpicture}
~~~
\begin{tikzpicture}[scale=0.4]
\begin{scope}
\draw [dashed] (0,3) grid (4,4) ;
\draw [blue] (0,0) -- (0,3) ;
\draw [red] (0,3) -- (4,3) ;
\draw [blue] (4,3) -- (4,4) ;
\end{scope}
\end{tikzpicture}
\end{center}
\caption{\label{fig_card_4}Partitions when $\binom{w}{ab} = 4$}
\end{figure}

%% file: partition_surface5.tex
\begin{figure}[!ht]  
\begin{center}
\begin{tikzpicture}[scale=0.4]
\begin{scope}
\draw [dashed] (0,0) grid (1,5) ;
\draw [red] (0,0) -- (1,0) ;
\draw [blue] (1,0) -- (1,5) ;
\draw [red] (1,5) -- (5,5) ;
\end{scope}
\end{tikzpicture}
~~~
\begin{tikzpicture}[scale=0.4] 
\begin{scope}
\draw [dashed] (0,1) grid (1,5) ;
\draw [dashed] (1,4) grid (2,5) ;
\draw [blue] (0,0) -- (0,1) ;
\draw [red] (0,1) -- (1,1) ;
\draw [blue] (1,1) -- (1,4) ;
\draw [red] (1,4) -- (2,4) ;
\draw [blue] (2,4) -- (2,5) ;
\draw [red] (2,5) -- (5,5) ;
\end{scope}
\end{tikzpicture}
~~~
\begin{tikzpicture}[scale=0.4] 
\begin{scope}
\draw [dashed] (0,2) grid (1,5) ;
\draw [dashed] (1,3) grid (2,5) ;
\draw [blue] (0,0) -- (0,2) ;
\draw [red] (0,2) -- (1,2) ;
\draw [blue] (1,2) -- (1,3) ;
\draw [red] (1,3) -- (2,3) ;
\draw [blue] (2,3) -- (2,5) ;
\draw [red] (2,5) -- (5,5) ;
\end{scope}
\end{tikzpicture}
~~~
\begin{tikzpicture}[scale=0.4] 
\begin{scope}
\draw [dashed] (0,2) grid (1,5) ;
\draw [dashed] (1,4) grid (3,5) ;
\draw [blue] (0,0) -- (0,2) ;
\draw [red] (0,2) -- (1,2) ;
\draw [blue] (1,2) -- (1,4) ;
\draw [red] (1,4) -- (3,4) ;
\draw [blue] (3,4) -- (3,5) ;
\draw [red] (3,5) -- (5,5) ;
\end{scope}
\end{tikzpicture}
~~~
\begin{tikzpicture}[scale=0.4] 
\begin{scope}
\draw (0,5.3) ;
\draw [dashed] (0,3) grid (2,5) ;
\draw [dashed] (2,4) grid (3,5) ;
\draw [blue] (0,0) -- (0,3) ;
\draw [red] (0,3) -- (2,3) ;
\draw [blue] (2,3) -- (2,4) ;
\draw [red] (2,4) -- (3,4) ;
\draw [blue] (3,4) -- (3,5) ;
\draw [red] (3,5) -- (5,5) ;
\end{scope}
\end{tikzpicture}
~~~
\begin{tikzpicture}[scale=0.4] 
\begin{scope}
\draw (0,5.3) ;
\draw [dashed] (0,3) grid (1,5) ;
\draw [dashed] (1,4) grid (4,5) ;
\draw [blue] (0,0) -- (0,3) ;
\draw [red] (0,3) -- (1,3) ;
\draw [blue] (1,3) -- (1,4) ;
\draw [red] (1,4) -- (4,4) ;
\draw [blue] (4,4) -- (4,5) ;
\draw [red] (4,5) -- (5,5) ;
\end{scope}
\end{tikzpicture}
~~~
\begin{tikzpicture}[scale=0.4] 
\begin{scope}
\draw (0,5.3) ;
\draw [dashed] (0,4) grid (5,5) ;
\draw [blue] (0,0) -- (0,4) ;
\draw [red] (0,4) -- (5,4) ;
\draw [blue] (5,4) -- (5,5) ;
\end{scope}
\end{tikzpicture}
\end{center}
\caption{\label{fig_card_5}Partitions when $\binom{w}{ab} = 5$}
\end{figure}